\documentclass[a4paper,UKenglish,cleveref, autoref, thm-restate]{lipics-v2021}

\usepackage{tabularx}
\usepackage{environ}
\usepackage{lipsum}
\usepackage{amsfonts}
\usepackage{cite}
\usepackage{xspace}
\usepackage{array}
\usepackage{xparse}
\usepackage{amsmath}
\usepackage{thm-restate}
\usepackage{xcolor}
\newcommand{\lipItem}[1]{\textcolor{lipicsGray}{\sffamily\bfseries\upshape\mathversion{bold}#1}}

\pdfoutput=1 
\hideLIPIcs  

\hypersetup{
    colorlinks, linkcolor={dark-blue},
    citecolor={mid-green}, urlcolor={dark-blue}}

\title{Constant congestion linkages in polynomially strong digraphs in polynomial time} 

\author{Raul Lopes}
{LIRMM, Université de Montpellier, Montpellier, France}
{raul.wayne@gmail.com}
{https://orcid.org/0000-0002-7487-3475}
{French project ELIT (ANR-20-CE48-0008-01).}

\author{Ignasi Sau}
{LIRMM, Université de Montpellier, Montpellier, France}
{ignasi.sau@lirmm.fr}
{https://orcid.org/0000-0002-8981-9287}
{French project ELIT (ANR-20-CE48-0008-01).}

\authorrunning{R. Lopes and I. Sau}

\Copyright{ }



\category{} 

\relatedversion{} 

\nolinenumbers 

\usepackage{tikz}

\usetikzlibrary{mindmap,positioning}
\usetikzlibrary{graphs}
\usetikzlibrary[graphs]
\usetikzlibrary{patterns}
\usetikzlibrary{arrows,automata,positioning}
\usetikzlibrary{decorations.pathreplacing}
\usetikzlibrary{decorations}
\usetikzlibrary{decorations.pathmorphing}
\usetikzlibrary{shapes.geometric}
\usetikzlibrary{quotes}
\usetikzlibrary{shapes.callouts}
\usetikzlibrary{arrows.meta}
\usetikzlibrary{calc}
\usetikzlibrary{decorations}
\usetikzlibrary{fit}

\tikzset{faded/.style={gray,very thin}}
\tikzset{vertex/.style={draw,circle,minimum size=5pt,inner sep=0pt}}
\tikzset{novertex/.style={circle,minimum size=5pt,inner sep=0pt}}
\tikzset{blackvertex/.style={draw,circle,minimum size=5pt,inner sep=0pt, fill=black}}
\tikzset{redvertex/.style={draw,circle,minimum size=5pt,inner sep=0pt, fill=red}}
\tikzset{redvertexfaded/.style={draw,circle,faded,minimum size=5pt,inner sep=0pt, fill=red!50}}
\tikzset{greenvertex/.style={draw,circle,minimum size=5pt,inner sep=0pt, fill=green}}
\tikzset{greenvertexfaded/.style={draw,circle,faded,minimum size=5pt,inner sep=0pt, fill=green!50}}
\tikzset{bluevertex/.style={draw,circle,minimum size=5pt,inner sep=0pt, fill=blue}}
\tikzset{bluevertexfaded/.style={draw,circle,faded,minimum size=5pt,inner sep=0pt, fill=blue!50}}
\tikzset{yellowvertex/.style={draw,circle,minimum size=5pt,inner sep=0pt, fill=yellow}}
\tikzset{yellowvertexfaded/.style={draw,circle,faded,minimum size=5pt,inner sep=0pt, fill=yellow!50}}
\tikzset{arrow/.style={-{Latex[scale=1]},shorten >= 0pt}}

\tikzset{edge/.style = {->,> = latex'}}

\tikzset{snake it/.style={decorate, decoration=snake}}

\EventEditors{John Q. Open and Joan R. Access}
\EventNoEds{2}
\EventLongTitle{42nd Conference on Very Important Topics (CVIT 2016)}
\EventShortTitle{IPEC 2024?}
\EventAcronym{CVIT}
\EventYear{2016}
\EventDate{December 24--27, 2016}
\EventLocation{Little Whinging, United Kingdom}
\EventLogo{}
\SeriesVolume{42}
\ArticleNo{23}

\newcommand{\dout}{{\sf deg}^+\xspace}
\newcommand{\din}{{\sf deg}^-\xspace}

\newcommand{\Nin}{{\sf N}^-\xspace}
\newcommand{\Nout}{{\sf N}^+\xspace}

\newcommand{\dtw}{{\sf dtw}\xspace}
\newcommand{\wlink}{{\sf wlink}\xspace}
\newcommand{\Int}{{\sf Int}\space}

\newcommand{\order}{\text{{\sf ord}}\xspace}
\newcommand{\occur}{\textsf{{\sf oc}}\xspace}

\newcommand{\ddpc}[1]{#1-\textsc{DDP}}

\newcommand{\masarik}{Masa\v{r}\'{\i}k\xspace}

\DeclareMathOperator{\Ocal}{\mathcal{O}}

\newcommand{\FPT}{\textsf{FPT}\xspace}
\newcommand{\XP}{\textsf{XP}\xspace}
\newcommand{\NP}{\textsf{NP}\xspace}
\newcommand{\Wh}{\textsf{W}[1]\xspace}

\newcommand{\poly}{\textsf{poly}\xspace}

\newcommand{\vbl}{\textsf{vbl}\xspace}
\newcommand{\Prob}{\textsf{Pr}\xspace}

\newcommand{\myParagraph}[1]{\medskip\noindent\textbf{#1}}

\definecolor{mid-green}{rgb}{0.15,0.65,0.15}
\definecolor{dark-green}{rgb}{0.15,0.25,0.15}
\definecolor{dark-red}{rgb}{0.7,0.15,0.15}
\definecolor{dark-blue}{rgb}{0.15,0.15,0.9}
\definecolor{medium-blue}{rgb}{0,0,0.5}
\definecolor{gray}{rgb}{0.5,0.5,0.5}
\definecolor{color-Ig}{rgb}{0.15,0.7,0.15}
\definecolor{darkmagenta}{rgb}{0.30, 0.0, 0.30}

\hypersetup{
    colorlinks, linkcolor={dark-blue},
    citecolor={mid-green}, urlcolor={dark-blue}
  }

\begin{document}

\maketitle

\begin{abstract}
Given positive integers $k$ and $c$, we say that a digraph $D$ is \emph{$(k,c)$-linked} if for every pair of ordered sets $\{s_1, \ldots, s_k\}$ and $\{t_1, \ldots, t_k\}$ of vertices of $D$, there are paths $P_1, \ldots, P_k$ such that for $i \in [k]$ each $P_i$ is a path from $s_i$ to $t_i$ and every vertex of $D$ appears in at most $c$ of those paths.
A classical result by Thomassen [Combinatorica, 1991] states that, for every fixed $k \geq 2$, there is no integer $p$ such that every $p$-strong digraph is $(k,1)$-linked.

Edwards et al. [ESA, 2017] showed that every digraph $D$ with directed treewidth at least some function $f(k)$ contains a large bramble of congestion $2$.
Then, they showed that every $(36k^3 + 2k)$-strong digraph containing a bramble of congestion $2$ and size roughly $188k^3$ is $(k,2)$-linked.
Since the directed treewidth of a digraph has to be at least its strong connectivity, this implies that there is a function $L(k)$ such that every $L(k)$-strong digraph is $(k,2)$-linked.
The result by Edwards et al. was improved by Campos et al. [ESA, 2023], who showed that any $k$-strong digraph containing a bramble of size at least $2k(c\cdot k -c + 2) + c(k-1)$ and congestion $c$ is $(k,c)$-linked.
Regarding how to find the bramble, although the given bound on $f(k)$ is very large, \masarik et al. [SIDMA, 2022] showed that directed treewidth $\mathcal{O}(k^{48}\log^{13} k)$ suffices if the congestion is relaxed to $8$.
In this article, we first show how to drop the dependence on $c$, for even $c$, on the size of the bramble that is needed in the work of Campos et al. [ESA, 2023].
Then, by making two local changes in the proof of \masarik et al. [SIDMA, 2022] we show how to construct in polynomial time a bramble of size $k$ and congestion $8$ assuming that a large obstruction to directed treewidth (namely, a path system) is given.
Applying those two results, we show that there is polynomial function $g(k)$ such that every $g(k)$-strong digraph is $(k,8)$-linked.
\end{abstract}

\section{Introduction}
In the $k$-\textsc{Directed Disjoint Paths} ($k$-\textsc{DDP}) problem, we are given a digraph $D$ and a set of $k$ pairs of vertices $s_i, t_i \in V(D)$ with $i \in [k]$ (the \emph{terminals}), and the goal is to decide if $D$ contains a set of pairwise vertex-disjoint paths (henceforth abbreviated to \emph{disjoint}) linking each $s_i$ to its pair $t_i$.
A solution to this problem is called a \emph{linkage}.
The $k$-\textsc{DDP} problem is one of the most fundamental problems in digraphs and, in contrast with the undirected version, shown to be \FPT with parameter $k$ by Robertson and Seymour~\cite{robertson1995graph}, Fortune et al.~\cite{fortune1980directed} showed that $k$-\textsc{DDP} is already \NP-complete for $k = 2$.
In addition, Slivkins~\cite{slivkins2010parameterized} showed that this problem is \Wh-hard with parameter $k$ even when restricted to acyclic digraphs (DAGs).

On the positive side, $k$-\textsc{DDP} retains some degree of tractability for restricted classes of digraphs.
For instance, Fortune et al.~\cite{fortune1980directed} showed that it is \XP with parameter $k$ in DAGs, which was later extended to digraphs of bounded \emph{directed treewidth} by Johnson et al.~\cite{Johnson2001},  by presenting an \XP algorithm with both $k$ and the directed treewidth as parameters.

The undirected and directed versions of the problem exhibit a very different behaviour with respect, in particular, to the (directed) treewidth of the input (di)graph. Indeed, the \FPT algorithm for $k$-\textsc{Disjoint Paths} in undirected graphs by Robertson and Seymour~\cite{robertson1995graph} shows that, essentially, the only instances that need to be solved are the ones whose input graph $G$ has bounded treewidth.
Indeed, they show that if the treewidth of $G$ is at least some computable function of $k$, then one can find in $G$ a vertex that can be safely deleted without changing the answer of the problem.
This \emph{irrelevant vertex} is extracted either from a large enough clique minor, or from deep inside a so-called \emph{flat wall} (roughly, a grid-like subgraph that behaves like a planar graph with respect to the routing of paths through its boundary), both structures acting as obstructions to bounded treewidth.

Such a \emph{win-win} strategy is not expected to work for $k$-\textsc{DDP}: since it is already \NP-complete for $k=2$~\cite{fortune1980directed} and \XP with parameters $k$ and directed treewidth~\cite{Johnson2001}, there is no function $f(k)$ that guarantees the existence of an irrelevant vertex that can be efficiently found in digraphs of directed treewidth at least $f(k)$.
More strongly, Fomin and Pilipzcuk~\cite{doi:10.1137/1.9781611973105.29} showed an example of a tournament $T$ with directed treewidth (in fact, even directed pathwidth) $|V(T)|/2$ that does not contain \textsl{any} irrelevant vertex.
Additionally, strong connectivity is also not helpful: Thomassen~\cite{Thomassen1991} showed that, even for fixed $k = 2$, the $k$-\textsc{DDP} problem remains \NP-complete in $p$-strongly connected digraphs (henceforth shortened to \emph{$p$-strong} digraphs) for any integer $p \geq 1$.
Thus, in order to cope with the intractability of $k$-\textsc{DDP}, the focus often shifts from restricting the class of digraphs to relaxing the problem.

The most common relaxation of $k$-\textsc{DDP} allows some degree of  \emph{\sl vertex congestion} for the solutions.
Namely, for integers $k$ and  $c$, an input of the $k$-\textsc{Directed} $c$-\textsc{Congested Disjoint Paths} ($(k,c)$-\textsc{DDP} for short) problem is the same as in $k$-\textsc{DDP}, but each vertex is allowed to appear in at most $c$ of the paths linking the terminals.
Notice that if $c \geq k$ then the problem can be solved trivially by simply testing for connectivity between the pairs of terminals, and thus we may assume that $k > c$.
Amiri et al.~\cite{AKHOONDIANAMIRI2019105836} showed how to modify existing results to prove that $(k,c)$-\textsc{DDP} with parameter $k$ remains \Wh-hard in DAGs and \XP in digraphs of bounded directed treewidth.
Despite considerable effort, it is open whether $(k,c)$-\textsc{DDP} is \XP with parameter $k$ in general digraphs for every fixed value of $c \geq 2$.
This fact, together with the fact that this problem is \XP for $c \geq 1$ in digraphs of bounded directed treewidth, justifies the growing interest in better understanding the behaviour of $(k,c)$-\textsc{DDP} in the presence of obstructions to bounded directed treewidth, such as \emph{brambles} (which we define later in this section and also in~\autoref{def:brambles-digraphs}).

In this direction, Edwards et al.~\cite{Edwards2017} showed conditions ensuring that {\sl every} instance of $(k,2)$-\textsc{DDP} in a digraph $D$ is positive. 
With our notation (see \autoref{def:k-linked-digraphs}), in that case we say that $D$ is \emph{\emph{$(k,c)$}-linked} for $c = 2$.
In particular, they showed that there is a (very large) function $L(k)$ such that every $L(k)$-strong digraph is $(k,2)$-linked.
We briefly discuss their result in the following paragraph.

A \emph{bramble} $\mathcal{B}$ in a digraph $D$ is a collection of strongly connected subgraphs of $D$, called \emph{bags}, such that for any $B,B'\in \mathcal{B}$ the digraph induced by $V(B) \cup V(B')$ is also strong.
The \emph{order} of $\mathcal{B}$ is the minimum size of a vertex set of $D$ intersecting all its elements.
For an integer $c \geq 1$, we say that $\mathcal{B}$ has \emph{congestion} $c$ if every vertex of $D$ appears in at most $c$ bags of $\mathcal{B}$.
Edwards et al.~\cite{Edwards2017} showed how to construct, in polynomial time, a bramble of order $k$ and congestion $2$ from a bramble of order $f(k)$, albeit with a very large exponential dependency on $k$.
This tool is a fundamental part of the main contribution of their paper, where they show that every $(36k^3 + 2k)$-strong digraph containing a bramble of congestion $2$ and size $188k^3$ is $(k,2)$-linked.
Notice that since a vertex can only hit $c$ elements of a bramble with congestion $c$, the order and size of each such bramble is linearly related to its size.
They also show how to find the desired set of paths in polynomial time assuming that the bags of the bramble are given.
Since one can always find a bramble of order $t$ in digraphs of directed treewidth at least $3t-1$ (see~\cite{Campos2022} and \cite[Theorem 9.4.4]{bang-jensen2018classes}), they conclude that there is a function  $g(k)$ such that every $(36k^3 + 2k)$-strong digraph of directed treewidth at least $g(k)$ is $(k,2)$-linked.
Finally, as the directed treewidth of a $t$-strong digraph $D$ is at least $t$, they conclude that there is a function $L(k)$ such that every $L(k)$-strong digraph is $(k,2)$-linked.

The result by Edwards et al.~\cite{Edwards2017} was recently improved and generalized by Campos et al.~\cite{Campos2022}.
For integers $k,c \geq 1$, they showed that every $k$-strong digraph containing a bramble of size $2k(c \cdot k - c + 2) + c(k-1)$ and congestion at most $c$ is $(k,c)$-linked, and that the paths can be found in polynomial time assuming that the bags of the bramble are given.
Applying this result and the construction of brambles of congestion $2$ of~\cite{Edwards2017}, they conclude that every $k$-strong digraph of directed treewidth at least $g(k)$ is $(k,c)$-linked.
The exponential dependence of $k$ in $g(k)$ appears as a consequence of the construction given in \cite{Edwards2017}. \masarik et al.~\cite{Masarik2022} showed that a polynomial bound for the existence of the bramble is possible if the congestion is relaxed to $8$, although their result is not presented in an algorithmic way.

\myParagraph{Our results and techniques.} Our contributions are twofold.
First, for an even integer $c \geq 2$, we show how to drop the dependence on $c$ on the size of a bramble of congestion $c$ that is used to solve $(k,c)$-\textsc{DDP} in $k$-strong digraphs in the result of Campos et al.~\cite{campos_et_al:LIPIcs.ESA.2023.30}.
Namely, in \autoref{theorem:linkedness-brambles-constant-congestion}  we show that, for $k \geq 1$, if $D$ is a $k$-strong and we are given the bags of a bramble $\mathcal{B}$ of size at least $4k^2 + (2k-1)$, then any instance of $(k, 2\lceil c/2 \rceil)$-\textsc{DDP} is positive and a solution can be found in time $\mathcal{O}(k^4\cdot n^2)$.
This is done by a simple reduction to the $(k,2)$-linked case.
In short, we construct from $\mathcal{B}$ a bramble $\mathcal{B}'$ of congestion $2$ and apply the result of Campos et al.~\cite{campos_et_al:LIPIcs.ESA.2023.30}.
We start from $\mathcal{B}$ and identify the vertices appearing in more than $2$ bags of $\mathcal{B}$.
Then, we make an appropriate number of copies of such vertices, depending on how many bags they appear in, and distribute those copies between the bags of $\mathcal{B}'$ in a way that ensures the desired congestion of $\mathcal{B}'$.
Since $D$ is $k$-strong, the terminals of any instance of $(k,c)$-\textsc{DDP} are ``well connected'' to $\mathcal{B}'$, and thus we can apply \cite[Theorem 16]{campos_et_al:LIPIcs.ESA.2023.30} (also stated as \autoref{proposition:solution_or_small_separator} in this article).

Our second and main contribution considers the result by \masarik et al.~\cite{Masarik2022}, where it is shown that there is an integer $t_k = \mathcal{O}(k^{48} \log^{13}k)$ such that every digraph of directed treewidth at least $t_k$ contains a bramble of size at least $k$ and congestion at most $8$.
We show how to obtain a polynomial-time algorithm from their proof (assuming that a \emph{path system}, as defined in \autoref{sec:preliminaries}, is given) by applying two local changes to it, since many parts of the proof are readily translatable into polynomial-time algorithms.

The local changes that we apply are about finding clique minors in graphs of large average degree, as done by Thomasson~\cite{Thomason1984} and Kostochka~\cite{Kostochka1984} (cf.~\cite[Lemma 2.4]{Masarik2022}), and an application of the \emph{Lovász Local Lemma}, henceforth abbreviated as LLL.
To find the desired clique minor in polynomial time, it suffices to replace the result used in~\cite{Masarik2022} by the algorithmic construction by Dujmovi\'c~\cite[Theorem 1.1]{Dujmovic2013}.

In what follows we sketch the main ingredients of our approach.
The proof of \masarik et al.~\cite{Masarik2022} follows a series of steps refining a previously obtained object into another one, starting from a path system (see \autoref{def:path-system}).
At many stages of the proof, the intersection graph of pairs of sets of linkages is tested for its degeneracy.
In case it is sufficiently large as a function of $k$, the desired bramble is found through an application of the results of~\cite{Thomason1984,Kostochka1984} (or, in our case, the algorithm of~\cite{Dujmovic2013}).
Otherwise, more subtle arguments are needed to distinguish between sparse cases.
In any case, the construction often relies on a classical application of the LLL to find multicoloured independent sets in graphs that have bounded degeneracy.
The original proof of the LLL by Lovász and Erdős~\cite{LovaszErdos1975}, however, is not constructive.
From this point forward, we have two options to obtain an algorithmic version of a result by Reed and Wood~\cite{REED2012374} appearing as~\cite[Lemma 2.5]{Masarik2022}.
In both cases, the running time of the algorithm is split as the sum of two parts, one of which depends only on the chosen approach to deal with the application of the LLL which, in turn, is applied on a graph whose size depends only on $k$ and on the choice of a constant $\alpha \geq 1$. We remark that this dependence is likely to be the dominant one in the running time of the obtained polynomial-time algorithm.

For the first, we can apply a result by Harris~\cite[Theorem 1.5]{Harris2022} to adapt the proof of \cite{Masarik2022} into a polynomial-time algorithm while maintaining the same asymptotic bound for the directed treewidth of the input digraph\footnote{This approach was suggested to us by David G. Harris after seeing the first version of this article on arXiv, containing only the second approach mentioned here. We would like to thank him for pointing out the existence of \autoref{proposition:Harris-independent-set} and for the helpful comments on how to apply it to obtain \autoref{theorem:strongly-connected-implies-linkedness-using-harris}.}.
The downside of this approach is that, according to Harris~\cite{Harris.personal.communicaton.2024}, in this case a precise bound on the running time of the LLL-dependent part of the algorithm cannot be given without significant technical work.
Thus, for the second option we apply an algorithmic version of the LLL by Chandrasekaran et al.~\cite{Chandrasekaran.Navin.Bernhard.09}, which requires more work, to obtain an upper bound on the running time of the  LLL-dependent part of the obtained algorithm.

The usage of the algorithmic version of the LLL by Chandrasekaran et al.~\cite{Chandrasekaran.Navin.Bernhard.09} has an impact on how large the directed treewidth must be to ensure that the bramble can be found.
We show that for every $\alpha > 1$ there is an integer $t'_k = \mathcal{O}\left((k \cdot \log k^2)^{\poly(\alpha)}\right)$ such that every digraph with directed treewidth at least $t'_k$ has a bramble of congestion $8$ and size $k$.
The choice of $\alpha$ has a large impact in the running time of the algorithm: $\alpha$ is directly proportional to $t'_k$ and, up to a certain point, the LLL-dependent algorithm runs faster as it grows (see the discussion at the end of \autoref{subsection:running-time-analysis}).
We remark that our contribution does not require any major changes to the proof by \masarik et al.~\cite{Masarik2022}: we follow the exact same steps of their proof, but adjusting the bounds appropriately for applying our result derived from an algorithmic version of the LLL.

As a direct consequence of our two main contributions, we are able to provide a polynomial function $f(k, \alpha)$ such that every $f(k, \alpha)$-strong digraph $D$ is $(k,8)$-linked.
Compared to the work by Edwards et al.~\cite{Edwards2017} we gain quantitatively by dropping the requirement on the strong connectivity to a polynomial function and lose qualitatively since the allowed congestion increases to $8$ from $2$.
Additionally, as in~\cite{Edwards2017}, given an instance of $(k,c)$-\textsc{DDP} in $D$ a solution can be found in polynomial time.

It is open whether  it is possible to improve the $\mathcal{O}(k^{48} \log^{13}k)$ bound from~\cite{Masarik2022} to obtain brambles of size $k$ and fixed congestion, and whether it is possible to obtain brambles of smaller congestion if a larger dependence on $k$ is allowed.
An improvement on the bound from~\cite{Masarik2022} implies an improvement to $f(k, \alpha)$ defined above.
If a large bramble of even congestion $c$ can be found, then \autoref{theorem:linkedness-brambles-constant-congestion} can be applied without any quality loss on the obtained paths.
If a bramble of congestion at most $6$ is obtainable, then again \autoref{theorem:linkedness-brambles-constant-congestion} can be applied to conclude $(k,c')$-linkedness of $k$-strong digraphs with $c' \leq 6$.

\myParagraph{Organization.}
In \autoref{sec:preliminaries} we give some preliminaries in graphs, digraphs, and parameterized complexity, and formally state our contributions and known results about brambles and linkages.
In \autoref{section:from-large-congestion-to-two} we prove our first contribution.
In \autoref{sec:finding-the-bramble} we introduce the necessary tools to obtain a polynomial-time algorithm from~\cite{Masarik2022}, prove our second contribution, and analyse the running time of the given algorithm.

\section{Preliminaries}\label{sec:preliminaries}
In this section we provide some basic preliminaries on graphs and digraphs, parameterized complexity, brambles, linkages, path systems, and $(k,c)$-linked digraphs.
\subsection{Graphs and digraphs}
We refer the reader to~\cite{Bondy2008} for basic background on graph theory, and recall here only some basic definitions. For a graph $G = (V,E)$, directed or not, and a set $X \subseteq V(G)$, we write $G \setminus X$ for the graph resulting from the deletion of $X$ from $G$.
If $e$ is an edge of a directed or undirected graph with \emph{endpoints} $u$ and $v$, we may refer to $e$ as $(u,v)$.
We also allow for loops and multiple edges.

The \emph{in-degree} $\din_D(v)$ (resp. \emph{out-degree} $\dout_D(v)$) of a vertex $v$ in a digraph $D$ is the number of edges reaching (resp. leaving) $v$ in $D$.
The \emph{in-neighbourhood} $\Nin_D(v)$ of $v$ is the set $\{u \in V(D) \mid (u,v) \in E(G)\}$, and the \emph{out-neighbourhood} $\Nout_D(v)$ of $v$ is the set $\{u \in V(D) \mid (v,u) \in E(G)\}$.
We say that $u$ is an \emph{in-neighbour} of $v$ if $u \in \Nin_D(v)$ and that $u$ is an \emph{out-neighbour} of $v$ if $u \in \Nout_D(v)$.

A \emph{walk} in a digraph $D$ is an alternating sequence $W$ of vertices and edges that starts and ends with a vertex, and such that for every edge $(u,v)$ in the walk, vertex $u$ (resp. vertex $v$) is the element right before (resp. right after) edge $(u,v)$ in $W$.
If the first vertex in a walk is $u$ and the last one is $v$, then we say this is a \emph{walk from $u$ to $v$}.
A \emph{path} is a digraph containing exactly a walk that contains all of its vertices and edges without repetition.
If $P$ is a path with $V(P) = \{v_1, \ldots, v_k\}$ and $E(P) = \{(v_i, v_{i+1}) \mid i \in [k-1]\}$, we say that $v_1$ is the \emph{first} vertex of $P$, that $v_k$ is the \emph{last} vertex of $P$, and for $i \in [k-1]$ we say that $v_{i+1}$ is the \emph{successor in $P$} of $v_{i}$.
All paths mentioned henceforth, unless stated otherwise, are considered to be directed.

A digraph $D$ is \emph{strongly connected} (or simply \emph{strong}) if, for every pair of vertices $u,v \in V(D)$, there is a walk from $u$ to $v$ and a walk from $v$ to $u$ in $D$.
A \emph{separator} of $D$ is a set $S \subsetneq V(D)$ such that $D \setminus S$ is not strongly connected.
If $|V(D)| \geq k+1$ and $k$ is the minimum size of a separator of $D$, we say that $D$ is \emph{$k$-strongly connected}.
A \emph{strong component} of $D$ is a maximal induced subgraph of $D$ that is strongly connected.

By \emph{contracting} an edge $e = (u,v)$ of a simple graph $G$, we generate a simple graph $G'$ starting from $G \setminus \{u,v\}$ then adding a new vertex $v_e$ adjacent too all neighbours of $u$ and $v$ in $G$, excluding $u$ and $v$.
A graph $H$ is a \emph{minor} of $G$ if $H$ can be obtained from $G$ by a sequence of vertex removals, edge removals, and edge contractions.
The definition naturally extends to digraphs.

For a positive integer $k$, a graph $G$ is \emph{$k$-degenerate} if every induced subgraph of $G$ contains  a vertex of degree at most $k$.
The \emph{degeneracy} of a graph $G$ is the least $\ell$ such that $G$ is $\ell$-degenerate.
We say that $G$ is \emph{$k$-partite} if there is a partition of $V(G)$ into $k$ sets such that every edge of $G$ has its endpoints in distinct sets of the partition.
We denote by $[k]$ the set containing every integer $i$ such that $1 \leq i \leq k$.

We also use Menger's Theorem.
\begin{theorem}[{Menger's Theorem}~\cite{Menger1927}] \label{thm:Menger}
Let $G$ be a digraph and $A,B \subseteq V(D)$.
The maximum size of a collection of disjoint $A \to B$ paths is equal to the minimum size of an $(A,B)$-separator.
\end{theorem}

\subsection{Parameterized complexity}

We refer the reader to~\cite{Downey2013,Cygan2015} for basic background on parameterized complexity, and we recall here only the definitions used in this article. A \emph{parameterized problem} is a language $L \subseteq \Sigma^* \times \mathbb{N}$.  For an instance $I=(x,k) \in \Sigma^* \times \mathbb{N}$, $k$ is called the \emph{parameter}.

A parameterized problem $L$ is \emph{fixed-parameter tractable} ({\sf FPT}) if there exists an algorithm $\mathcal{A}$, a computable function $f$, and a constant $c$ such that given an instance $I=(x,k)$, $\mathcal{A}$   (called an {\sf FPT} \emph{algorithm}) correctly decides whether $I \in L$ in time bounded by $f(k) \cdot |I|^c$. For instance, the \textsc{Vertex Cover} problem parameterized by the size of the solution is {\sf FPT}.

A parameterized problem $L$ is in {\sf XP} if there exists an algorithm $\mathcal{A}$ and two computable functions $f$ and $g$ such that given an instance $I=(x,k)$, $\mathcal{A}$  (called an {\sf XP} \emph{algorithm}) correctly decides whether $I \in L$ in time bounded by $f(k) \cdot |I|^{g(k)}$. For instance,  the \textsc{Clique} problem parameterized by the size of the solution is in  {\sf XP}.

Within parameterized problems, the class {\sf W}[1] may be seen as the parameterized equivalent to the class {\sf NP} of classical decision problems. Without entering into details (see~\cite{Downey2013,Cygan2015} for the formal definitions), a parameterized problem being {\sf W}[1]-\emph{hard} can be seen as a strong evidence that this problem is {\sl not} {\sf FPT}.
The canonical example of {\sf W}[1]-hard problem is \textsc{Clique}  parameterized by the size of the solution.

\subsection{Brambles, linkages, and path systems}
In this section we define the most fundamental objects studied in this paper.
From this point forward, and unless stated otherwise, we refer to oriented edges only, $D$ will always stand for a digraph, and $G$ for an undirected graph.
Moreover, when $D$ is the input digraph of some algorithm we set $n = |V(D)|$.
\begin{definition}[Brambles in digraphs]\label{def:brambles-digraphs}
A \emph{bramble} $\mathcal{B} = \{B_1, \ldots, B_\ell\}$ in a digraph $D$ is a family of strongly connected subgraphs of $D$ such that if $\{B, B'\} \subseteq \mathcal{B}$ then $V(B) \cap V(B') \neq \emptyset$ or there are edges in $D$ from $V(B)$ to $V(B')$ and from $V(B')$ to $V(B)$.
A \emph{hitting set} of a bramble $\mathcal{B}$ is a set $X \subseteq V(D)$ such that $X \cap V(B) \neq \emptyset$ for all $B \in \mathcal{B}$. The \emph{order}  of a bramble $\mathcal{B}$, denoted by $\order(\mathcal{B})$,  is the minimum size of a hitting set of $\mathcal{B}$.
We also say that the sets $B_i \in \mathcal{B}$ are the \emph{bags} of $\mathcal{B}$.
For an integer $c \geq 1$ we say that $\mathcal{B}$ has \emph{congestion $c$} if every vertex of $D$ occurs in at most $c$ bags of the bramble.
\end{definition}
Clearly every bramble $\mathcal{B}$ of congestion $c$ has order at least $\lceil |\mathcal{B}|/c \rceil$ since, by definition, every vertex of a hitting set of $\mathcal{B}$ can hit at most $c$ bags of $\mathcal{B}$.

Originally, the result of Campos et al.~\cite{Campos2022} builds a \emph{haven} of order $k$ in digraphs of directed treewidth at least $k$.
Since one can easily construct a bramble of order $\lfloor k/2\rfloor$ from a haven of order $k$ and generate a haven of order $k$ from a bramble of order $k$ (see~\cite[Chapter 6]{Matthias2014} for example), we refrain from defining havens and state the next proposition with relation to brambles only, noting that the original statement refers to havens.
We also do not need the definition of directed treewidth and thus it is also skipped.
The aforementioned \XP algorithm~\cite{Johnson2001} was improved to an {\sf FPT} algorithm, with the same approximation ratio, by Campos et al.~\cite{Campos2022}.
We remark that a sketch of a proof of an \FPT algorithm is also available in \cite[Theorem 9.4.4]{bang-jensen2018classes}, but with a slightly worse approximation ratio.
\begin{proposition}[Campos et al.~\cite{Campos2022}, bramble version]
\label{proposition:directed-treewidth-bramble}
Let $D$ be a digraph and $k$ be a non-negative integer.
There is an algorithm running in time $2^{\Ocal(k \log k)}\cdot n^{\Ocal(1)}$ that correctly decides that $D$ has directed treewidth at most $3k-2$ or outputs a bramble of order $k$ in $D$.
\end{proposition}

Both constructions of brambles with constant congestion from Edwards et al.~\cite{Edwards2017} and  Masa\v{r}\'{\i}k et al.~\cite{Masarik2022} start from a highly connected system of paths.
In order to formally define this system, we adopt the following notation.
\begin{definition}[Linkages]
Let $D$ be a digraph and $A,B \subseteq V(D)$ with $A \neq B$.
A \emph{linkage from $A$ to $B$} in $D$, or an \emph{$(A, B)$-linkage}, is a set of of pairwise vertex-disjoint paths from $A$ to $B$.
\end{definition}

As is this case for brambles, \emph{well-linked sets} are also known obstructions for bounded directed treewidth and are also involved in the definition of a path system.
\begin{definition}[Well-linked sets]\label{def:well-linked-sets-in-digraphs}
Let $D$ be a digraph and $A \subseteq V(D)$.
We say that $A$ is \emph{well-linked} in $D$ if, for every ordered pair  $(X,Y)$ of disjoint sets $X,Y \subseteq A$ with $|X| = |Y|$, there is an $(A,B)$-linkage of size $|X|$ in $D$ (not necessarily contained in $A$).
The \emph{order} of a well-linked set $A$ is $|A|$.
We denote by $\wlink(D)$ the size of a largest well-linked set in $D$.
\end{definition}

\begin{definition}[Path system]\label{def:path-system}
Let $D$ be a digraph and $a,b$ be two positive integers.
An \emph{$(a,b)$-path system} is a tuple $\mathcal{S}$ with $\mathcal{S} = (\mathcal{P}, \mathcal{L}, \mathcal{A})$ where
\begin{itemize}
\item $\mathcal{A} = \{A_i^{\textsf{\emph{in}}}, A_i^{\textsf{\emph{out}}} \mid i \in [a]\}$ where each $A_i^{\textsf{\emph{in}}}$ and each $A_i^{\textsf{\emph{out}}}$ is a well-linked set of order $b$;

\item $\mathcal{L}$ is a collection $\{L_{i,j} \mid i,j \in [a] \text{ with } i \neq j\}$ of linkages where each $L_{i,j}$ is a linkage of size $b$ from $A_i^{\textsf{\emph{out}}}$ to $A_j^{\textsf{\emph{in}}}$; and

\item $\mathcal{P}$ is a sequence $P_1, \ldots, P_a$ of pairwise vertex-disjoint paths such that, for all $i \in [a]$, $V(P_i) \supseteq A_i^{\textsf{\emph{in}}} \cup A_i^{\textsf{\emph{out}}}$ and every vertex in $A_i^{\textsf{\emph{in}}}$ appears in $P_i$ before any vertex of $A_i^{\textsf{\emph{out}}}$.

\end{itemize}
\end{definition}
Although there is a lot to unpack in the definition of path systems, it is not hard to visualize; see \autoref{fig:path_system} for an illustration.

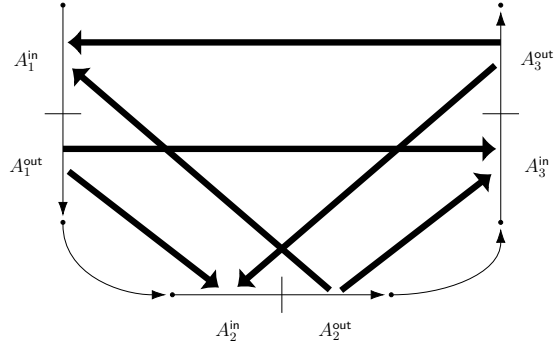
\begin{figure}[h!]
  \centering
  \scalebox{.6}{
  \begin{tikzpicture}[scale=.8]

  \node[blackvertex, scale=.5] (u1) at (0,0) {};
  \node[blackvertex, scale=.5] (u0) at ($(u1) + (0,6)$) {};

  \node[blackvertex, scale=.5] (v0) at (12,0) {};
  \node[blackvertex, scale=.5] (v1) at ($(v0) + (0,6)$) {};

  \node[blackvertex, scale=.5] (w0) at (3,-2) {};
  \node[blackvertex, scale=.5] (w1) at ($(w0) + (6,0)$) {};

  \draw[-{Latex[length=3mm, width=2mm]}, shorten >= .1cm] (u0) -- (u1) node [midway] (um) {} node [pos = .25] (u0um) {} node [pos = .75] (umu1) {};
  \draw ($(um) + (-0.5, 0)$) -- ($(um) + (0.5, 0)$);
  \node (Ain) at ($(u0um) + (-1, 0)$) {\Large$A_1^{\textsf{in}}$};
  \node (Ain) at ($(umu1) + (-1, 0)$) {\Large$A_1^{\textsf{out}}$};

  \draw[-{Latex[length=3mm, width=2mm]}, shorten >= .1cm] (v0) -- (v1) node [midway] (vm) {} node [pos = .25] (v0vm) {} node [pos = .75] (vmv1) {};
  \draw ($(vm) + (-0.5, 0)$) -- ($(vm) + (0.5, 0)$);
  \node (Ain) at ($(v0vm) + (1, 0)$) {\Large$A_3^{\textsf{in}}$};
  \node (Ain) at ($(vmv1) + (1, 0)$) {\Large$A_3^{\textsf{out}}$};

  \draw[-{Latex[length=3mm, width=2mm]}, shorten >= .1cm] (w0) -- (w1) node [midway] (wm) {} node [pos = .25] (w0wm) {} node [pos = .75] (wmw1) {};
  \draw ($(wm) + (0, -0.5)$) -- ($(wm) + (0, 0.5)$);
  \node (Ain) at ($(w0wm) + (0, -1)$) {\Large$A_2^{\textsf{in}}$};
  \node (Ain) at ($(wmw1) + (0, -1)$) {\Large$A_2^{\textsf{out}}$};

  \draw[-{Latex[length=3mm, width=6mm]}, line width=4, shorten >= .1cm] ($(umu1) + (0,.5)$) -- ($(v0vm) + (0, .5)$);
  \draw[-{Latex[length=3mm, width=6mm]}, line width=4, shorten >= .1cm] (umu1) -- (w0wm);

  \draw[-{Latex[length=3mm, width=6mm]}, line width=4, shorten >= .1cm] ($(vmv1) + (0, .5)$) -- ($(u0um) + (0,.5)$);
  \draw[-{Latex[length=3mm, width=6mm]}, line width=4, shorten >= .1cm] (vmv1) -- (w0wm);

  \draw[-{Latex[length=3mm, width=6mm]}, line width=4, shorten >= .1cm] (wmw1) -- (u0um);
  \draw[-{Latex[length=3mm, width=6mm]}, line width=4, shorten >= .1cm] (wmw1) -- (v0vm);

  \draw[-{Latex[length=3mm, width=2mm]}, shorten >= .1cm] (u1) to [out = -90, in = 180] (w0);
  \draw[-{Latex[length=3mm, width=2mm]}, shorten >= .1cm] (w1) to [out = 0, in = -90] (v0);

  \end{tikzpicture}%
  }
  \caption{An $(a,b)$-path system with $a =3$. A thick edge denotes a linkage of size $b$ from a set $A_i^{\text{out}}$ to a set $A_j^{\text{in}}$, with $i \neq j$.}
  \label{fig:path_system}
\end{figure}

As shown by Kawarabayashi and Kreutzer~\cite{Kawarabayashi2015}, it is easy to construct, in polynomial time, a large path system from a path $P$ intersecting all elements of a bramble $\mathcal{B}$ of order roughly $k^2$ and a well-linked set $A \subseteq V(P)$ with $|A| \geq k$.
Finding $P, \mathcal{B}$, and $A$ in a digraph with large directed treewidth is no trivial matter.
The first algorithm that does so is (implicitly) given in~\cite{Kawarabayashi2015} and runs in \XP time parameterized by $k$.
This was later improved in~\cite[Theorem 2.23]{Campos2022} to an algorithm in time $2^{\Ocal(k^2 \log k)}\cdot n^{\Ocal(1)}$, and a better bound on $\order(\mathcal{B})$ is also given.
Thus plugging together the results from~\cite{Kawarabayashi2015} and~\cite{Campos2022} we obtain the following.

\begin{proposition}\label{proposition:campos-dtw-to-path-system}
Let $a,b$ be non-negative integers.
There is a constant $c$ such that every digraph $D$ with $\dtw(D) \geq c\cdot a^2 \cdot b^2$ contains an $(a,b)$-path system.
Additionally, such a path system can be constructed in {\sf FPT} time parameterized by $\max\{a,b\}$.
\end{proposition}

\subsection{Constant congestion brambles and $(k,c)$-linked digraphs}

Edwards et al.~\cite[Theorem 14]{Edwards2017}\footnote{The precise statement considers that a path $P$ and a well-linked set $A \subseteq V(P)$ are given, and the first step of their proof is to build a path system from those.} showed how to construct in polynomial time a bramble of congestion $2$ and size $k$ from a $(g(k), g'(k))$-path system, where both $g(k)$ and $g'(k)$ are exponential towers.
Plugging together this result, the aforementioned result by Campos et al.~\cite[Theorem 2.23]{Campos2022}, and \autoref{proposition:directed-treewidth-bramble}, we conclude that in \FPT time with parameter $k$ one can find a bramble of congestion $2$ and size at least $k$ in digraphs with directed treewidth at least $f(k)$, albeit with an exponential bound on $f(k)$.
Masa\v{r}\'{\i}k et al.~\cite{Masarik2022} showed the bound on $f(k)$ can be made polynomial if the congestion on the bramble is relaxed to $8$, but no algorithm executing this construction was given.
\begin{proposition}[Masa\v{r}\'{\i}k et al.~\cite{Masarik2022}]\label{proposition:bramble-masarik-congestion-8}
For every integer $k \geq 1$ there is an integer $t = \Ocal(k^{48} \log^{13}k)$ such that every digraph of directed treewidth at least $t$ contains a bramble of congestion $8$ and size at least $k$.
\end{proposition}
The proof of \autoref{proposition:bramble-masarik-congestion-8} skips directed treewidth and assumes that an $(a,b)$-path system is given, where both $a$ and $b$ depend polynomially on $k$ (see \autoref{proposition:campos-dtw-to-path-system}).
\autoref{proposition:bramble-masarik-congestion-8} can be adapted into a polynomial-time algorithm by applying the  result by Harris~\cite[Theorem 1.5]{Harris2022} (cf. \autoref{proposition:Harris-independent-set}) without any loss on the asymptotic bound of $t_k$.
However, according to the author~\cite{Harris.personal.communicaton.2024}, significant technical work is needed to precise the running time of the algorithm.

To provide a version of \autoref{proposition:bramble-masarik-congestion-8} with a bound on the running time, in \autoref{sec:finding-the-bramble} we show how to adapt their construction into a polynomial-time algorithm that finds a bramble of congestion $8$ and size at least $k$, assuming that an $(a,b)$-path system is given with $a$ and $b$ also depending polynomially on $k$, but with a worse bound than that of \autoref{proposition:bramble-masarik-congestion-8}.
This loss is mainly due to an application of a polynomial-time and deterministic algorithm for the Lovász Local Lemma.
More precisely, we obtain the following.
\begin{restatable}{theorem}{congestionBramblePolyTime}\label{theorem:constant-congestion-brambles-poly-time}
For every integer $k \geq 1$ there are integers $t$ and $\ell$ depending only polynomially on $k$ such that every digraph $D$ containing a $(t, \ell)$-path system contains  a bramble of size $t$ and congestion $8$.
Moreover, a bramble with the desired properties can be found in polynomial time assuming that the $(t, \ell)$-path system is given.
\end{restatable}
The values of $t$ and $\ell$ will be precised later. A statement in the same shape of \autoref{proposition:bramble-masarik-congestion-8} is given in \autoref{theorem:strongly-connected-implies-linkedness}.
The requirement of \autoref{theorem:constant-congestion-brambles-poly-time} is not much different than the requirement of \autoref{proposition:bramble-masarik-congestion-8} since in their proof the authors use the directed treewidth only to obtain a path system.

We remark that it is easy to construct path systems in strongly connected digraphs.
For completeness, we include a proof of this statement.
\begin{proposition}\label{proposition:path-systems-strong-digraphs}
Let $k\geq 1$ be an integer and $D$ be a $2k^2$-strong digraph.
There is an algorithm running in time $\Ocal(n^2)$ that finds a $(k,k)$-path system in $D$.
\end{proposition}
\begin{proof}
Since $D$ is $2k^2$-strongly connected we know that $\dout(v) \geq 2k^2$ holds for all $v \in V(D)$.
Thus every inclusion-wise vertex maximal path $P$ of $D$ has $|V(P)| \geq 2k^2$.
Indeed, otherwise the last vertex $v$ of $P$ has out-neighbours only in $V(P)$ and this contradicts the fact that $\dout(v) \geq 2k^2$.

Assume now that $|V(P)| = 2k^2$.
We split $P$ into $k$ vertex-disjoint subpaths $P_1, P_2, \ldots, P_k$, appearing in $P$ in this order and all with size $2k$.
For $i \in [k]$ define $A^{\text{in}}_i$ to be the set formed by the first $k$ vertices of $P_i$ and $A^{\text{out}}_i$ to be formed by the last $k$ vertices of $P_i$.
Since $D$ is $2k^2$-strongly connected, each $A^{\text{in}}_i$ and each $A^{\text{out}}_i$ is well-linked and there is a linkage of size $k$ from each $A^{\text{out}}_i$ to each $A^{\text{in}}_j$ with $i \neq j$.

Clearly the algorithm runs in time $\Ocal(n^2)$ since $P$ can be constructed in time $\Ocal(n^2)$, and the result follows.
\end{proof}

We now discuss the relation between constant congestion brambles and linkages in digraphs.
\begin{definition}[$(k,c)$-linked digraphs]\label{def:k-linked-digraphs}
We say that a digraph $D$ is \emph{$(k,c)$-linked} if for every disjoint pair of ordered sets $S,T \subseteq V(D)$ with $S = \{s_i \mid i \in [k]\}$, and $T = \{t_i \mid i \in [k]\}$ there is a collection of paths $\{P_i \mid i \in [k]\}$ where $P_i$ is a path from $s_i$ to $t_i$ in $D$ and such that each vertex of $D$ appears in at most $c$ paths of the collection.
\end{definition}
Notice that $(k,1)$-linkedness is a property stronger than being $k$-strongly connected: the latter only cares about $k$ vertex-disjoint paths from $S$ to $T$, while the former also asks the paths to start and end on prescribed vertices of $S$ and $T$.
In other words, if $D$ is $(k,c)$-linked then every instance of \ddpc{$(k,c)$} in $D$ is positive (although finding a solution may still be a non-trivial task).
Edwards et al.~\cite{Edwards2017} showed that every $(36k^3 + 2k)$-strong digraph containing a bramble of congestion $2$ and size $188k^3$ is $(k,2)$-linked.
This was later improved by Campos et al.~\cite{campos_et_al:LIPIcs.ESA.2023.30}.
\begin{proposition}[Campos et al.~\cite{campos_et_al:LIPIcs.ESA.2023.30}]\label{proposition:campos-bramble-linkages}
For all $k,c \geq 2$ every $k$-strong digraph $D$ containing a bramble of congestion $c$ and size at least $2k(c\cdot k - c + 2) + c(k-1)$ is $(k,c)$-linked.
Moreover, given an instance $(D, S, T)$ of \ddpc{$(k,c)$}, there is an algorithm that finds a solution in time $\Ocal(k^4 \cdot n^2)$.
\end{proposition}
Notice that the requested size of the bramble remains quadratic in $k$ only for constant $c$.
In \autoref{section:from-large-congestion-to-two} we obtain a better version of \autoref{proposition:campos-bramble-linkages} by showing how to eliminate the dependence on $c$ from the requested size of the bramble.
Namely, we prove the following.
\begin{restatable}{theorem}{linkednessBramblesCongestion}\label{theorem:linkedness-brambles-constant-congestion}
For all $k \geq 1$ and $c \geq 2$ every $k$-strong digraph $D$ containing a bramble $\mathcal{B}$ of congestion $c$ and size at least $4k^2 + 2(k-1)$ is $(k,2\lceil c/2\rceil)$-linked.
Moreover, given an instance $(D, S, T)$ of \ddpc{$(k,2\lceil c/2 \rceil)$} and $\mathcal{B}$, there is an algorithm that finds a solution in time $\Ocal(k^4 \cdot n^2)$ assuming that the bags of $\mathcal{B}$ are given.
\end{restatable}

Pipelining \cref{theorem:constant-congestion-brambles-poly-time,theorem:linkedness-brambles-constant-congestion} and adjusting the bounds accordingly, we obtain the following.
\begin{theorem}\label{theorem:strongly-connected-implies-linkedness}
For all $\alpha > 1$ and integer $k \geq 1$ there is an integer
\[t_k = \mathcal{O}\left(k^{8 + 16\alpha + 32\alpha^2 + 40\alpha^3}\cdot (\log k^2)^{1 + 2\alpha + 4\alpha^2 + 6\alpha^3} \right)\]
such that every $k$-strong digraph $D$ with $\dtw(D) \geq t_k$ is $(k,8)$-linked.
Additionally, any given instance $(D, S, T)$ of \ddpc{$(k,8)$} can be solved in time $\poly(k,\alpha) \cdot n^2$, where the degree of the polynomial depending on $k$ and $\alpha$ comes from $\autoref{theorem:poly-LLL-independent-sets}$.
\end{theorem}
The running time of \autoref{theorem:strongly-connected-implies-linkedness} depends on \autoref{theorem:poly-LLL-independent-sets} and  is further discussed in the end of \autoref{sec:algorithmic-construction}.
Alternatively, using \autoref{proposition:Harris-independent-set} we obtain a better bound on $t_k$ (the same as \autoref{theorem:strongly-connected-implies-linkedness} with $\alpha = 1$) but the degree of $\poly(k, \alpha)$ is unknown.
\begin{theorem}\label{theorem:strongly-connected-implies-linkedness-using-harris}
For all $\alpha > 1$ and integer $k \geq 1$ there is an integer
\[t_k = \mathcal{O}\left(k^{96}\cdot (\log k^2)^{13} \right)\]
such that every $k$-strong digraph $D$ with $\dtw(D) \geq t_k$ is $(k,8)$-linked.
Additionally, any given instance $(D, S, T)$ of \ddpc{$(k,8)$} can be solved in polynomial time.
\end{theorem}

Finally, applying \autoref{proposition:path-systems-strong-digraphs}, \autoref{theorem:strongly-connected-implies-linkedness}, \and \autoref{theorem:strongly-connected-implies-linkedness-using-harris} we prove the following.

\begin{corollary}
For all $\alpha \geq 1$ and integers $k \geq 1$ there is a polynomial function $g(k, \alpha)$ such that every $g(k, \alpha)$-strong digraph is $(k,8)$-linked.
\end{corollary}
\section{From congestion $c$ to congestion $2$}\label{section:from-large-congestion-to-two}
\autoref{proposition:campos-bramble-linkages} is a consequence of another result from the same article~\cite{campos_et_al:LIPIcs.ESA.2023.30} stating that, given an instance $(D, S, T)$ of \ddpc{$(k,c)$} and a sufficiently large bramble $\mathcal{B}$ of congestion $c$ in $D$, we can either find a positive solution to $(D, S, T)$ or a separator of size at most $k-1$ intersecting all paths from $S$ to the bags of $\mathcal{B}$ or all paths from the bags of $\mathcal{B}$ to $T$.
Thus if $D$ is $k$-strong such a separator cannot exist and \autoref{proposition:campos-bramble-linkages} is obtained.
\begin{proposition}[Campos et al.~\cite{campos_et_al:LIPIcs.ESA.2023.30}]\label{proposition:solution_or_small_separator}
Let $k,c$ be integers with $k,c \geq 2$ and $g(k,c) = 2k(c\cdot k - c + 2) + c(k-1)$.
Let $D$ be a digraph, assume that we are given the bags of a bramble $\mathcal{B}$ of congestion $c$ and size at least $g(k,c)$, and sets $S,T \subseteq V(D)$ with $S = \{s_1, \ldots, s_k\}$ and $T = \{t_1, \ldots, t_k\}$.
Then in time $\Ocal(k^4 \cdot n^2)$ one can either
\begin{enumerate}
  \item find a set of paths $\{P_1, \ldots, P_k\}$ in $D$ such that each $P_i$ with $i \in [k]$ is a path from $s_i$ to $t_i$ and each vertex of $D$ appears in at most $c$ of these paths, or
  \item find a $\mathcal{B}^* \subseteq \mathcal{B}$ with $|\mathcal{B}^*| \geq g(k,c) - c(k-1)$ and a set $X \subseteq V(D)$ such that $|X| \leq k-1$, $X$ is disjoint from all bags of $\mathcal{B}^*$, and $X$ intersects all paths from $S$ to the bags of $\mathcal{B}^*$ or all paths from these bags to $T$.
\end{enumerate}
\end{proposition}

\autoref{proposition:solution_or_small_separator} gives us a little more freedom in the following proof, since it allows us to modify the given digraph into one that does not need to be $k$-strong.
To ensure that item \lipItem{2.} of \autoref{proposition:solution_or_small_separator} does not occur, we only need that no small separator appears between $S$ and a large part of the bramble $\mathcal{B}$, nor between a large part of $\mathcal{B}$  and $T$.
In what follows, by adding a \emph{copy} of a vertex $v \in V(D)$ to $D$ we mean the operation of adding to $D$ a new vertex $v'$ that has the same in- and out-neighbours as $v$.
For clarity, we restate \autoref{theorem:linkedness-brambles-constant-congestion} before giving its proof.

\linkednessBramblesCongestion*
\begin{proof}
Let $S,T$ be disjoint subsets of $V(D)$ with $S = \{s_i \mid i \in [k]\}$ and $T = \{t_i \mid i \in [k]\}$.
Let $\{B_i \mid i \in [r]\}$ be the bags of $\mathcal{B}$.
The goal is to construct an instance $(D', S', T')$ of \ddpc{$(k,2)$} such that we can apply \autoref{proposition:solution_or_small_separator} with a guarantee that only the first outcome is possible.
Then we translate the obtained solution to $(D', S', T')$ into the desired paths from $s_i$ to $t_i$ in $D$ and the result will follow. To construct $\mathcal{B}'$, the idea is to make copies of every vertex of $D$ occurring in at least $3$ bags of $\mathcal{B}$, and distribute these copies along the bags containing $v$ ensuring that the newly generated bramble has congestion $2$. Let us now describe how this construction is done.

Start with $D' = D$ and $\mathcal{B}' = \mathcal{B}$.
Fix a bijective mapping $\sigma : \mathcal{B} \to \mathcal{B}'$ where $\sigma(B) = B'$ if and only $V(B) = V(B')$.
We can assume that $\mathcal{B}$ does not contain two copies of the same bag since the congestion of $\mathcal{B}$ can only decrease when removing duplicate bags.
We remark that although the bags themselves of $\mathcal{B'}$ may change throughout the proof, the mapping $\sigma$ is defined here and remains unchanged from now on.
For every $v \in V(D)$ let $\occur(v) = |\{B \in \mathcal{B} \mid v \in V(B)\}|$.

For $i \in k$, add to $D'$ copies $s'_i$ of $s_i$ and $t'_i$ of $T_i$ and define $S' = \{s'_i \mid i \in [k]\}$ and $T' = \{t'_i \mid i \in [k]\}$.
Then, add to $D'$ an arc from $s'_i$ to $s_i$ and an arc from $t'_i$ to $t_i$.
Thus clearly every bag of $\mathcal{B}'$ is disjoint from $S' \cup T'$.

Let $X$ be the set containing all $v \in V(D)$ with $\occur(v) \geq 3$.
For each $v \in X$ denote by $B^v_1, \ldots, B^v_{\occur(v)}$ the elements of $\mathcal{B}$ containing $v$, define $\ell = \lceil\frac{1}{2} \occur(v) \rceil$ and add to $D'$ copies $v^2, \ldots, v^\ell$ of $v$.
Define $v^1 = v$.
Now, for each $i \in [\ell-1]$ we set
\begin{itemize}
 \item $V(\sigma(B^v_{2i-1})) = (V(B^v_{2i-1}) \setminus \{v\}) \cup \{v^i\}$, and
 \item $V(\sigma(B^v_{2i})) = (V(B^v_{2i}) \setminus \{v\}) \cup \{v^i\}$.
\end{itemize}
Set $V(\sigma(B^v_{2\ell-1})) = (V(B^v_{2\ell-1}) \setminus \{v\}) \cup \{v_i\}$ and, if $\occur(v)$ is even, set $V(\sigma(B^v_{2\ell})) = (V(B_{2\ell}) \setminus \{v\}) \cup \{v_i\}$.
For $v \in X$ let $X_v \subseteq V(D')$ contain  $v$ and the copies of $v$ that were added to $D'$.
Now, for all $v \in X$ we add to $D'$ all arcs in both directions between all pairs of vertices in $X_v$.
This serves two purposes: first it ensures that $\mathcal{B}'$ is indeed a bramble of $D'$, since the construction described so far may result in two bags of $\mathcal{B}$ becoming disjoint and without arcs in both directions in $D'$ after copying vertices and distributing them in the bags of $\mathcal{B}'$.
Second, it ensures that any path of $D'$ using a copy of a vertex $v \in X$ can access all bags of $\mathcal{B}'$ associated with bags of $\mathcal{B}$ containing $v$ using only copies of $v$.

Finally, for $i \in [k]$, add to $D'$ vertices $s'_i$ and $t'_i$ and define $S' = \{s'_i \mid i \in [k]\}$ and $T' = \{t'_i \mid i \in [k]\}$.
If $s_i \in X$, then add to $D'$ all arcs from $s'_i$  to $X_{s_i}$.
Otherwise, add to $D'$ arcs from $s'_i$ to $s_i$.
Similarly, if $t_i \in X$ add to $D'$ all arcs from from $X_{t_i}$ to $t'_i$ and, otherwise, add to $D'$ an arc from $t_i$ to $t'_i$.
See \autoref{fig:construction_bramble_transformation} for an illustration of this construction.
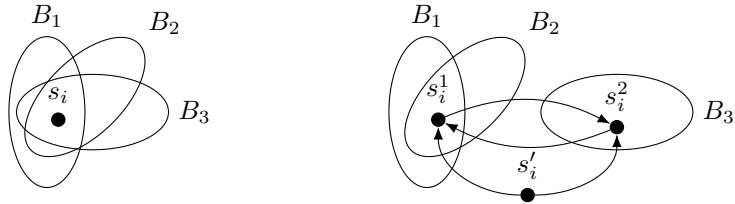
\begin{figure}[h]
\centering
\begin{tikzpicture}
\node[draw, ellipse, minimum width = 1cm, minimum height = 2cm, label = {$B_1$}] (B1) at (0,0) {};
\node[draw, ellipse, minimum width = 1cm, minimum height = 2cm, rotate = 90, label = -90:{$B_3$}] (B3) at (.6,0) {};
\node[draw, ellipse, minimum width = 1cm, minimum height = 2cm, rotate = -45, label = 90:{$B_2$}] (B2) at (.5,.2) {};
\node[blackvertex, label = 90:{$s_i$}] (s) at (.15,-.1) {};
\begin{scope}[xshift = 5cm]
\node[draw, ellipse, minimum width = 1cm, minimum height = 2cm, label = {$B_1$}] (B1) at (0,0) {};
\node[draw, ellipse, minimum width = 1cm, minimum height = 2cm, rotate = 90, label = -90:{$B_3$}] (B3) at (2.5,0) {};
\node[draw, ellipse, minimum width = 1cm, minimum height = 2cm, rotate = -45, label = 90:{$B_2$}] (B2) at (.5,.2) {};
\node[blackvertex, label = 90:{$s^1_i$}] (s) at (.15,-.1) {};
\node[blackvertex, label = 90:{$s^2_i$}] (s1) at ($(B3) + (0, -.2)$) {};
\draw[arrow] (s) to [bend left = 25] (s1);
\draw[arrow] (s1) to [bend left = 25] (s);

\node[blackvertex, label = 90:{$s'_i$}] (s2) at ($(s) + (1.175,-1)$) {};
\draw[arrow] (s2) to [out = 180, in = -90] (s);
\draw[arrow] (s2) to [out = 0, in = - 90] (s1);

\end{scope}
\end{tikzpicture}
\caption{Bramble transformation used in the proof of \autoref{theorem:linkedness-brambles-constant-congestion}. In the example, $s_i \in S$ and $\occur(s_i) = 3$.}
\label{fig:construction_bramble_transformation}
\end{figure}

We claim that if there is a solution $\{P'_1, \ldots, P'_k\}$ to the instance $(D', S', T')$ of \ddpc{$(k,2)$} then there is a collection of paths $\{P_1, \ldots, P_k\}$ in $D$ such that each $P_i$ with $i \in [k]$ is a path from $s_i$ to $t_i$ and every vertex of $D$ appears in at most $2\lceil c/2\rceil$ paths of the collection.
We can assume that no $P'_i$ contains two vertices of some $X_v$ since such a path can be shortened by keeping only the first occurrence of such a vertex in $P'_i$ and jumping from this vertex to the out-neighbour in $P'_i$ of the last occurrence of a vertex of $X_v$ in the path.

Now for each $P'_i$ we construct $P_i$ by simply starting from $s_i$ and following $P'_i$ substituting each occurrence of a vertex $u \in X_v$ by $v$.
Since $|X_v| \leq \left\lceil \frac{1}{2}\occur(v)\right\rceil$  and every vertex of $D'$ occurs in at most two paths of $\{P'_1, \ldots, P'_k\}$, it immediately follows that every vertex of $D$ occurs in at most $2\lceil c/2\rceil$ paths among $P_1, \ldots, P_k$.

Since $|\mathcal{B}'| = |\mathcal{B}|$ and $D$ is $k$-strong, the arcs leaving $S'$ and entering $T'$ in $D'$ guarantee that there are at least $k$ disjoint paths from $S'$ to the bags of any $\mathcal{B}^* \subseteq \mathcal{B}'$ with size at least $k$, and at least $k$ disjoint paths from the bags of any such $\mathcal{B}^*$ to $T'$.
From \autoref{thm:Menger}, it follows that we can apply \autoref{proposition:solution_or_small_separator} with input $D', \mathcal{B}', S'$, and $T'$ with a guarantee that only the first outcome is possible and thus we conclude that $D$ is $(k, \lceil c/2 \rceil)$-linked.

The running time follows from the running time of \autoref{proposition:solution_or_small_separator} and by observing that $\mathcal{B}'$ can be constructed in time $\Ocal(|V(D)| \cdot |\mathcal{B}|\cdot c)$.
\end{proof}

\section{Finding the bramble in polynomial time}\label{sec:finding-the-bramble}

To prove \autoref{theorem:constant-congestion-brambles-poly-time}, we need to introduce results about algorithmic versions of the Lovász Local Lemma and on how to find clique minors in dense graphs.

\subsection{The tools}
Let $g(t)$ be the minimum integer $d$ such that every graph with average degree at least $d$ contains a $K_t$-minor.
Mader~\cite{Mader1967} showed that $g(t)$ is well-defined and that $g(t) \geq 2^{t-2}$.
The best bound of $g(t) = \Theta(t \sqrt{\log t})$ was given by Thomason~\cite{Thomason1984} and Kostochka~\cite{Kostochka1984}.
Dujmovi\'{c} et al.~\cite{Dujmovic2013} showed an algorithm that, under the same conditions, finds a $K_t$-minor in the given graph.
Since every graph of degeneracy at least $d$  has a non-empty subgraph with average degree (in fact, even minimum degree) at least $d$, we can rewrite the original statement of this result using degeneracy instead of average degree.
\begin{proposition}[Dujmovi\'{c} et al.~\cite{Dujmovic2013}, degeneracy version]\label{proposition:degeneracy-clique-minor}
For any sufficiently large  $t$, there is a constant $c$ and an algorithm running in time $\Ocal(n)$ that, given a graph $G$ with degeneracy at least $c \cdot t\sqrt{\log t}$, finds a $K_t$-minor in $G$.
\end{proposition}

Informally, the \emph{Lovász Local Lemma} says that, given a set $\mathcal{X}$ of ``bad'' events in a probability space, if the probability that an event in $\mathcal{X}$ occurs is not very high and they are mostly independent from each other, then there is a configuration where no event in $\mathcal{X}$ occurs.
More precisely, assume that each event in $\mathcal{X}$ occurs with probability at most $p$ and is independent of all but at most $d$ other events in $\mathcal{X}$.
Then if $e \cdot p(d+1) \leq 1$ with positive probability no event in $\mathcal{X}$ occurs.
This statement is known as the \emph{symmetric LLL} and the original version is described on more general terms.

Although the LLL is a powerful tool to find good combinatorial structures in complex environments, the original proof~\cite{LovaszErdos1975} is not constructive.
Thus, even if we know that some desirable structure indeed exists, finding it might still be no trivial matter.
Fortunately, many advances were made in designing algorithms for the LLL.
Moser and Tardos~\cite{10.1145/1667053.1667060} showed a constructive proof of the LLL that ends in \emph{expected} polynomial time, and a deterministic polynomial-time algorithm for fixed $d$.
For the application of Masa\v{r}\'{\i}k et al.~\cite{Masarik2022} and, consequently, ours, the goal is to use the LLL to find multicoloured independent sets.
To bound $d$, they exploit the fact that the input graph has bounded degeneracy.
Their reliance on the LLL is one of the reasons why their proof cannot be immediately translated into a constructive algorithm, and the method used to bound $d$ is the reason why we cannot rely on the algorithm by Moser and Tardos~\cite{10.1145/1667053.1667060} for fixed $d$.
Thus we use the algorithmic version of the LLL by Chandrasekaran et al.~\cite[Theorem 5]{Chandrasekaran.Navin.Bernhard.09} (a more general result has been recently proven by Harris~\cite{Harris2022}), which requests a slightly stronger relation over $p$ and $d$.

We could've used the symmetric form of~\cite[Theorem 5]{Chandrasekaran.Navin.Bernhard.09}, which can be easily derived from the original statement, to show that there is a polynomial-time algorithm that finds a bramble of congestion at most $8$ in digraphs of large directed treewidth. However, in order to give a precise bound on the running time, we need to work with the full generality of their result.

\subsection{Finding multicoloured independent sets}\label{subsection:multicolored-independent-sets}
As mentioned in the previous sections, we have two options to find the multicolored independent sets in our setup.
We can apply the following result by Harris~\cite[Theorem 1.5]{Harris2022} to obtain \autoref{theorem:strongly-connected-implies-linkedness-using-harris}, which has a better bound for $t_k$ with the downside of not providing a precise bound of the running time.
\begin{proposition}[Harris~\cite{Harris2022}]\label{proposition:Harris-independent-set}
Let $G$ be a digraph with maximum degree $\Delta$ and $\mathcal{V}$ be a partition of $V(G)$ into sets $V_1, \ldots, V_r$ of size exactly $t$. Let $\varepsilon > 0$ be an arbitrary constant and suppose that $t \geq (2+\varepsilon)\Delta$. Then there is a polynomial-time algorithm that finds in $G$ an independent set $\{v_1, \ldots, v_r\}$ such that $v_i \in V_i$ for each $i \in [r]$.
\end{proposition}

It is easy to apply \autoref{proposition:Harris-independent-set} under the assumption that $G[V_i \cup V_j]$ is $b$-degenerate for any distinct $i,j \in [r]$. It suffices to notice that, since there are at most $2t(r-1)b$ edges leaving each $V_i$ in this case, at least $t/2$ vertices of each such set have maximum degree at most $4t(r-1)b$. Thus one can apply \autoref{proposition:Harris-independent-set},  for example,  under the assumption that $t \geq 17(r-1)b$. 
Since this bound on $t$ is within a constant factor of the original conditions needed to apply~\cite[Lemma 2.5]{Masarik2022} (which asks that $t \geq 4e(r-1)b$), to obtain \autoref{theorem:strongly-connected-implies-linkedness} it suffices to analyse the running time of the remaining steps of the original construction, as done in the end of \autoref{sec:algorithmic-construction}.

To obtain \autoref{theorem:strongly-connected-implies-linkedness}, additional work is need to apply~\cite[Theorem 5]{Chandrasekaran.Navin.Bernhard.09}.
We need the following definitions.
Let $\mathcal{P} = \{P_1\, \ldots, P_n\}$ be a collection of mutually independent discrete random variables and, for $i \in [n]$, let $D_i$ be the finite domain of $P_i$.
Let $\mathcal{A} = \{A_1, \ldots, A_m\}$ be a collection of events each of which is determined by a subset of $\mathcal{X}$.
The \emph{variable set} of an event $A \in \mathcal{A}$ is the unique minimal $\mathcal{S} \subseteq \mathcal{P}$ that determines $A$ and is denoted by $\vbl(A)$.
We denote by $\Gamma(A)$ the set $\{B \in \mathcal{A} \mid \vbl(A) \cap \vbl(B) \neq \emptyset\}$.
Thus, $\Gamma(A)$ denotes the set of events of $\mathcal{A}$ which are \textsl{not} independent from $A$.

Let $x : \mathcal{A} \to (0,1)$ and
\begin{itemize}
  \item $x'(A) = x(a) \cdot \prod_{B\in \Gamma(A)} (1 - x(B))$,
  \item $D = \max_{P_i \in \mathcal{P}}|D_i|$,
  \item $M = \max \{n, 4m, 2\sum\frac{2|\vbl(A)|}{x'(A)}\cdot\frac{x(A)}{1-x(A)}, \max_{A \in \mathcal{A}}\frac{1}{x'(A)}\}$, and
  \item $w_{\text{min}} = \min_{A \in \mathcal{A}}{-\log x'(A)}$.
\end{itemize}
With those definitions \cite[Theorem 5]{Chandrasekaran.Navin.Bernhard.09} can be stated as follows.
The goal is to use $\mathcal{A}$ to describe the set of bad events that we wish to avoid.
Thus, we say that an evaluation of the random variables is a \emph{good evaluation} if no event in $\mathcal{A}$ occurs.

\begin{proposition}[Chandrasekaran et al.~\cite{Chandrasekaran.Navin.Bernhard.09}]\label{proposition:LLL-polynomial-time}
For $I \subseteq [n]$, let $t_C$ be the time needed to compute the conditional probability $\Prob(A \mid \forall i \in I, P_i = v_i)$, for any $A \in \mathcal{A}$ and any partial evaluation $\{v_i \in D_i \mid i \in I\}$.
Assume that there is an $\varepsilon \in (0,1)$ and an assignment of reals $x: \mathcal{A} \to (0,1)$ such that for all $A \in \mathcal{A}$
\begin{equation}\label{eq:poly-LLL-complete-conditions}
\Prob(A) \leq (x'(A))^{1+\varepsilon} = \left(x(A)  \prod_{B \in \Gamma(A)}(1-x(B)) \right)^{1+\varepsilon}.
\end{equation}
Then there is a deterministic algorithm that finds a good evaluation in time
\begin{equation}\label{poly-LLL-time}
\mathcal{O}\left(t_C \cdot \frac{DM^{3+2/\varepsilon} \cdot \log^2 M}{\varepsilon^2 \cdot w^2_{\text{min}}} \right).
\end{equation}
\end{proposition}

We now use \autoref{proposition:LLL-polynomial-time} to derive an algorithmic version of a result by Reed and Wood~\cite[Lemma 4.3]{REED2012374}.

\begin{theorem}\label{theorem:poly-LLL-independent-sets}
Let $r$ be a positive integer and $b$ a positive real.
Let $G$ be a graph and $\mathcal{V}$  be a partition of $V(G)$ into sets $V_1, \ldots, V_r$ of size exactly $t$.
Suppose that
\begin{equation}\label{eq:poly-LLL}
t^{1-\varepsilon} \geq \left(e \cdot 4b(r-1)\right)^{1+\varepsilon}
\end{equation}
for some $\varepsilon > 0$ and that $G[V_i \cup V_j]$ is $b$-degenerate for any distinct $i,j \in [r]$.
Then $G$ contains an independent set $\{v_1, \ldots, v_r\}$ such that $v_i \in V_i$ for each $i \in [r]$ that is found in time
\begin{equation}\label{eq:running-time-poly-LLL}
\mathcal{O}\left(r \cdot t \cdot (r^2\cdot b \cdot t)^{3+2/\varepsilon} \cdot \frac{\log^2(r^2 \cdot b \cdot t)}{\log^2(r \cdot b \cdot t)}\right).
\end{equation}
\end{theorem}
\begin{proof}
For each $i \in [r]$, let $P_i$ be a random variable that picks as value, uniformly at random, a vertex from $V_i$.
Then, each $v_i$ is chosen with probability $1/t$.
Let $F$ be the set of edges of $G$ with endpoints in distinct sets of $\mathcal{V}$.
By the degeneracy assumption, we conclude that $|F| \leq r^2 \cdot b \cdot 2t$.
Now, for each $(u,v) \in F$ if $A_{u,v}$ denotes the event that both endpoints of $e$ are taken then $\Prob(A_{u,v}) = 1/t^2$.
The goal is to find a good evaluation, that is, a choice of a vertex $v_i \in V_i$ for each $i \in [r]$, which avoids every event of the set $\mathcal{A} = \{A_{u,v} \mid (u,v) \in F\}$ of ``bad'' events.

By the assumption that each $G[V_i \cup V_j]$ is $b$-degenerate and since every $b$-degenerate graph with $\ell$ vertices has at most $b \cdot \ell$ edges, we conclude that, for each $i \in [r]$, the number of edges with exactly one endpoint in $V_i$ is at most $(r-1)b \cdot 2t$.
This, in turn, implies that $|\Gamma(A_{u,v})| \leq 4b(r-1)$ since, by definition, $A_{x,y} \in \Gamma(A_{u,v})$ if $\{x,y\} \cap \{u,v\} \neq \emptyset$ and $\{x,y\} \neq \{u,v\}$.
The intuition here is that, for $u \in V_i$ and $v \in V_j$ with $i \neq j$, the event $A_{u,v}$ is independent from every other event in $\mathcal{A}$ associated with an edge $e \in F$ with endpoints in sets of $\mathcal{V}$ distinct from both $V_i$ and $V_j$, and thus the degeneracy assumption implies that the bad events are reasonably independent from each other.
Defining $d = 4b(r-1)t$, we are ready to invoke \autoref{proposition:LLL-polynomial-time}.

Since every $\Prob(A_{u,v}) = 1/t^2$, to apply \autoref{proposition:LLL-polynomial-time} it suffices to choose $x(A_{u,v}) = 1/d$ for all $(A_{u,v} \in \mathcal{A})$.
Now, we want that
\[
  \frac{1}{t^2} \leq \left(x(A) \prod_{B \in \Gamma(A_{u,v})}(1-x(B))\right)^{1+\varepsilon}.
\]
Since $|\Gamma(A_{u,v})| \leq d$ and $(1 - 1/d)^d \leq 1/e$ we can rewrite the above inequality as
\[
  t^2 \geq (d \cdot e)^{1+\varepsilon},
\]
and thus the choice of $t$ guarantees that the conditions of \autoref{proposition:LLL-polynomial-time} are satisfied and a good evaluation is found in polynomial time.

It remains to analyse the running time of the algorithm and, to do this, we go over each of the dependencies one by one.
First, the time $t_C$ to compute the conditional probabilities as in the statement of the theorem depends only on $r$ since it suffices to check, for each $A_{u,v} \in \mathcal{A}$, if and which vertices of the sets of $\mathcal{V}$ containing $u$ and $v$ have been chosen.
Thus $t_C = \mathcal{O}(r)$.

By definition, $x'(A) = 1/(d \cdot e)$, $D = t$ since each set $V_i$ has size $t$, and $|\vbl(A_{u,v})| = 2$ since every event $A_{u,v}$ is determined by the two distinct random variables $P_i, P_j$ such that $u\in V_i$ and $v \in V_j$.
Let $m = |F|$.
Then
\[2\cdot\sum_{A_{u,v}\in \mathcal{A}}\frac{2|\vbl(A_{u,v})|}{x'(A_{u,v})} \cdot \frac{x(A_{u,v})}{1-x(A_{u,v})} = 2m\cdot\frac{4d\cdot e}{d-1},\]
which in turn implies that $M = \mathcal{O}(r\cdot d) = \mathcal{O}(r^2 \cdot b \cdot t) = \mathcal{O}(m)$ since $m \leq r^2 \cdot b \cdot 2t$ by the degeneracy assumption.
Finally, $w_{\text{min}} = \log d\cdot e$ implies that the running time of the algorithm is
\[
\mathcal{O}\left(r \cdot t \cdot m^{3+2/\varepsilon} \cdot \frac{\log^2(m)}{\log^2(d)}\right)
\]
and the result follows.
\end{proof}
If \autoref{eq:poly-LLL} holds with respect to some choice of $t, b, r$, and $\varepsilon$, we may say that \emph{$t$ satisfies the polynomial LLL conditions}.
Clearly, this is stronger than the classical LLL requirements (which, in this case, would imply $t \geq e\cdot2b\cdot(r-1)$) and thus whenever we satisfy the former, the existential claims in~\cite{Masarik2022} relying on the latter are still guaranteed to hold, as well as properties of the constructed objects which are obtained from the construction.
This allows us to skip a large part of the original analysis, focusing only on the parts that we need to adapt.

\subsection{The algorithmic construction}\label{sec:algorithmic-construction}
In this section we give an sketch of the proof of \autoref{theorem:constant-congestion-brambles-poly-time}.
The goal is to guide the reader through a high-level overview of the original proof by \masarik et al.~\cite{Masarik2022} while highlighting which parts and which bounds on the requested parameters need to be changed, and point out which steps have a significant impact on the running time.

The \emph{intersection graph} of a family of sets $\mathcal{A}$, denoted by $\Int(\mathcal{A})$, is the undirected graph with vertex set $\{A \mid A \in \mathcal{A}\}$ where two distinct vertices $A,A'$ are adjacent if and only if $A \cap A' \neq \emptyset$.
Given two linkages $L$ and $L'$ of  size $\ell$ and paths of a digraph $D$, clearly one can construct $\Int(L \cup L')$ in time $\Ocal(n \cdot \ell^2)$ by simply verifying intersections between paths of the linkages.
We say that a family $\mathcal{W}$ of closed walks in a digraph $D$ has \emph{congestion} $\alpha$ if every vertex of $D$ appears in at most $\alpha$ walks of $\mathcal{W}$.
The construction often mentions and uses an invariant of (closed) walks that the authors of \cite{Masarik2022} call \emph{overlap}.
Informally, the overlap is an upper bound on how many times a single vertex can appear in the (closed) walks, and this property is used to maintain control on the congestion of the obtained bramble in the end.

Let $a,b,d_1,d_2$, and $d_3$ be positive integers, to be precised later.
We follow the construction of~\cite{Masarik2022}, at each point justifying the choice of the values for $b,d_1, d_2$, and $d_3$, as they are needed, in order to apply the polynomial LLL given in \autoref{theorem:poly-LLL-independent-sets}.
The choice of $a$ is irrelevant for us, since it is only used when computing the size of the obtained bramble, and here we do not repeat these calculations.
Let $c_{a,b}$ be the constant of \autoref{proposition:campos-dtw-to-path-system} and assume that $D$ is a digraph with $\dtw(D) \geq c_{a,b}\cdot a^2 \cdot b^2$.
Thus $D$ contains an $(a,b)$-path system $\mathcal{(P, L, A)}$.
By \autoref{proposition:campos-dtw-to-path-system}, this step can be done in \FPT time with parameter $\max\{a,b\}$ (or in polynomial time if $D$ is $\max\{a^2,b^2\}$-strong applying \autoref{proposition:path-systems-strong-digraphs}), and we argue that every other step from this point forward can be done in polynomial time.
We remark that we make no attempt to optimize the choice of the following values, but only show that at least one good choice is possible for each of them.

Although we frequently mention \emph{threaded linkages} and \emph{untangled threaded linkages}, as introduced in~\cite{Masarik2022}, we do not define these objects since they are not necessary to understand the changes that we make to the proof of \cite{Masarik2022}.
The goal of these mentions is to guide the reader to which part of the original proof we are referring to.

Let $G$ be a graph with one vertex for each linkage $L_{i,j}$ of the path system $\mathcal{(P, L, A)}$ with $i \neq j$ and no edges.
Formally, $V = \{(i,j) \in [a] \times [a] \mid i \neq j\}$ and $G = (V, \emptyset)$.
The authors of~\cite{Masarik2022} apply \cite[Lemma 3.5]{Masarik2022} to $\mathcal{(P, L, A)}$ and thus associate each vertex of $G$ with a threaded linkage of size $b$, as defined in~\cite{Masarik2022}.
The goal now is to apply \cite[Lemma 3.6]{Masarik2022} to each of those objects to associate with each vertex of $G$  either a set of closed walks with additional properties, including small congestion, or with an untangled threaded linkage.
For that, we need that $b \geq x\cdot d + (d-1)$, for some choice of $x$ and $d$, in order to conclude that the obtained sets of closed walks have size $d$ and the untangled threaded linkages have size at least $x$.
Let $\alpha > 1$.
Choosing
\begin{align*}
b &= \left\lceil (e \cdot 4a^2 \cdot d^2_1)^{\alpha}\right\rceil,\\
x &= (e     \cdot 4a^2 \cdot d_1)^{\alpha} + 1, \text{ and}\\
d &= \frac{d_1^{\alpha}}{2^9 \cdot 5}
\end{align*}
we conclude that
\begin{align*}
x\cdot d + (d - 1) &= \left((e \cdot 4a^2 \cdot d_1)^{\alpha} + 1\right) \cdot \frac{d_1^{\alpha}}{2^9 \cdot 5} + \frac{d_1^{\alpha}}{2^9 \cdot 5} - 1\\
&= \frac{(e \cdot 4a^2 \cdot d_1)^{\alpha}\cdot d_1^{\alpha}}{2^9 \cdot 5} + \frac{d_1^{\alpha}}{2^9 \cdot 5} + \frac{d_1^{\alpha}}{2^9 \cdot 5} - 1\\
&= \frac{ (e \cdot 4a^2 \cdot d^2_1)^{\alpha}}{2^9 \cdot 5} + \frac{d_1^{\alpha}}{2^8\cdot 5} - 1 \\
&\leq (e \cdot 4a^2 \cdot d^2_1)^{\alpha} \leq b.
\end{align*}
Thus, the lemma is applicable.
The set of vertices of $G$ for which the first outcome is obtained, namely the closed walks, is named $Z$.
As a consequence of applying \cite[Lemma 3.6]{Masarik2022} to each $(i,j) \in Z$, we can also associate each such vertex with a \emph{threaded linkage} satisfying some additional properties.
Whether each $(i,j) \in Z$ is used to identify one of such threaded linkages or a closed walk depends on the application.
For now,
\begin{itemize}
  \item each $(i,j) \in Z$ is associated with a collection of closed walks $Z_{i,j}$ and with a threaded linkage $L'_{i,j}$, both of size $d_1^{\alpha}/(2^9 \cdot 5)$, and
  \item each $(i,j) \in V \setminus Z$ is associated with an untangled threaded linkage of size at least $(e \cdot 4a^2 \cdot d_1)^{\alpha} + 1$.
\end{itemize}
We remark that, so far, we have applied \cite[Lemmas 3.5 and 3.6]{Masarik2022} a total of $a(a-1)$ times each.

They now build two undirected graphs on top of $V$, each depending on the choices of $d_1$ and $d_2$ with $d_1 \geq d_2$.
These choices are done in order to satisfy the LLL conditions, and thus need to be adapted to satisfy the {\sl polynomial} LLL conditions.
Start with $E_1 = E_2 = \emptyset$.
For each pair $(i,j), (i', j') \in V$ and $i \in \{1,2\}$, add the pair to $E_i$ if $\Int(L'_{i,j} \cup L'_{i', j'})$ is {\sl not} $d_i$-degenerate and define $H_i = (V, E_i)$.
Finally, let $M_1$ be a maximum matching in $H_1 \setminus Z$ and $M_2$ be a maximum matching in $H_2$ excluding every edge with both endpoints in $Z$ plus vertices used by $M_1$.
For $i \in \{1,2\}$ let $V(M_i)$ be the set of vertices used by the edges of $M_i$.
It is shown~\cite[Claim 4.1]{Masarik2022} that at least one of these three cases occurs:
\begin{enumerate}
  \item $|V \setminus (V(M_1) \cup Z)| \geq 0.6|V|$,
  \item $|V(M_1) \cup V(M_2) \cup Z| \geq 0.6|V|$, or
  \item $|V \setminus V(M_2)| \geq 0.6|V|$.
\end{enumerate}
In each of those cases, the goal is to apply one of the two following winning schemes.

The algorithmic version of the \emph{sparse winning scenario} is obtained by simply following the original proof and applying~\autoref{proposition:degeneracy-clique-minor} instead of~\cite[cf. Lemma 2.4]{Masarik2022}.
\begin{lemma}[sparse winning scenario~{\cite[Lemma 3.2]{Masarik2022}}, algorithmic version]\label{lemma:sparse-winning-scenario}
There is an absolute constant $c$ with the following property.
Let $D$ be a digraph, let $a,b \geq 1$ be integers and $(P, \mathcal{L}, \mathcal{A})$ be an $(a,b)$-path system in $D$.
Let $\mathcal{I}$ be a subset of $[a] \times [a] \setminus \{(i,i) \mid i \in [a]\}$ with size at least $0.6a(a-1)$, and assume that for every $(i,j) \in \mathcal{I}$ there is a path $P_{i,j}$ from $A^{\textsf{\emph{out}}}_i$ to $A^{\textsf{\emph{in}}}_j$ such that $\{P_{i,j} \mid (i,j) \in \mathcal{I}\}$ has congestion at most $\alpha$.
Then we can find  in polynomial time a bramble in $D$ of congestion at most $2 + 2\alpha$ and size at least
\[c \left(\frac{\sqrt{a}}{\sqrt[4]{\log a}}\right).\]
\end{lemma}

Similarly, the algorithmic version of the \emph{sparse winning scenario (wrapped)} is obtained by following the original proof and adjusting the bounds to apply~\autoref{theorem:poly-LLL-independent-sets} instead of~\cite[Lemma 2.5]{Masarik2022}.
\begin{lemma}[sparse winning scenario, wrapped~{\cite[Lemma 3.3]{Masarik2022}}, algorithmic version]\label{lemma:sparse-winning-scenario-wrapped}
Let $c$ be the constant from \autoref{lemma:sparse-winning-scenario}.
Let $D$ be a digraph, let $a,b \geq 1$ be integers and $(P, \mathcal{L}, \mathcal{A})$ be an $(a,b)$-path system in $D$.
Let $\mathcal{I}$ be a subset of $[a] \times [a] \setminus \{(i,i) \mid i \in [a]\}$ with size at least $0.6a(a-1)$, and assume that there is an integer $d$ such that, for every distinct pair $(i,j),(i', j') \in \mathcal{I}$ the intersection graph $\Int(L_{i, j} \cup L_{i', j'})$ is $d$-degenerate.
If $b^{1-\varepsilon} \geq (e \cdot 4a^2 \cdot d)^{1+\varepsilon}$, then in polynomial time we can find in $D$ a bramble of congestion at most $4$ and size at least
\[c \left(\frac{\sqrt{a}}{\sqrt[4]{\log a}}\right).\]
\end{lemma}

In addition, for  case \lipItem{2.} another computation step is needed in the form of the \emph{Bowtie Lemma}~\cite[Lemma 3.7]{Masarik2022}, whose proof also naturally yields a polynomial-time algorithm, and the \emph{dense winning scenario}.
Again, the algorithmic version follows by incorporating~\autoref{proposition:degeneracy-clique-minor} to the original proof.
\begin{lemma}[dense winning scenario~{\cite[Lemma 3.1]{Masarik2022}}, algorithmic version]\label{lemma:dense-winning-scenario}
Let $c_t$ be the constant from \autoref{proposition:degeneracy-clique-minor}.
If a digraph $D$ contains a family $\mathcal{W}$ of closed walks of congestion at most $\alpha$ whose intersection graph is not $c_t \cdot d \cdot \sqrt{\log d}$ degenerate, then in polynomial time we can find in $D$ a bramble of congestion $\alpha$ and size $d$.
\end{lemma}

For case \lipItem{1.}, let $\mathcal{I} = V \setminus (V(M_1) \cup Z)$.
The goal is to apply \autoref{lemma:sparse-winning-scenario-wrapped} to a path system built from linkages associated with pairs $(i,j) \in \mathcal{I}$.
Since the pairs in $Z$ are excluded, every linkage in $\{L'_{i,j} \mid (i,j) \in \mathcal{I}\}$ has size at least $(e \cdot 4a^2 \cdot d_1)^{\alpha} + 1$.
Thus we need that
\begin{align*}
\left((e \cdot 4a^2\cdot d_1)^{\alpha}\right)^{1-\varepsilon} &\geq \left(e \cdot 4a^2 \cdot d_1\right)^{1+\varepsilon}\\
(e \cdot 4a^2 \cdot d_1)^{\alpha(1-\varepsilon)} &\geq (e \cdot 4a^2 \cdot d_1)^{1+\varepsilon}.
\end{align*}
Thus, it suffices to choose $\alpha$ and $\varepsilon$ such that
\begin{equation}
\alpha(1-\varepsilon) \geq 1 + \varepsilon,
\end{equation}
noticing that the running time of \autoref{theorem:poly-LLL-independent-sets} (and thus of \cref{theorem:linkedness-brambles-constant-congestion,theorem:strongly-connected-implies-linkedness}) is inversely proportional to $\varepsilon$, and the increase in the original dependency of $k$ needed for the construction is directly proportional to $\alpha$.
Since $\alpha$ has to be at least $1$ for the choices of $b$, $x$, and $d$ above to be feasible for the construction, and $\varepsilon > 0$ is needed to apply \autoref{theorem:poly-LLL-independent-sets}, we conclude that $0 < \varepsilon < 1$ and $\alpha > 1$.
Taking $\alpha$ and $\varepsilon$ satisfying those bounds and observing that $|\mathcal{I}| \geq 0.6|V| \geq 0.6a(a-1)$ (since in $V$ we excluded only the pairs $(i,i)$ with $i\in [a]$), we apply \autoref{lemma:sparse-winning-scenario-wrapped} and finish the construction.

In case \lipItem{2.}, the goal is to apply \autoref{lemma:sparse-winning-scenario} to a path system  built in three steps.
For each each $e \in M_1 \cup M_2$ with endpoints $(i,j)$ and $(i', j')$, the Bowtie Lemma~\cite[Lemma 3.7]{Masarik2022} is applied to $L_{i,j}$ and $L_{i', j'}$ to, in short, construct sets of closed walks with bounded congestion assuming that the intersection graph of the two linkages has sufficiently large degeneracy.
By the definitions of $H_1$ and $H_2$, for each $\ell \in \{1,2\}$ and $e \in M_{\ell}$ the intersection graph of the pair of threaded linkages associated with the endpoints of $e$ is not $d_\ell$-degenerate, each such edge is associated with a set $Z_e$ of closed walks of congestion at most $4$ with $|Z_e| \geq d_\ell/(2^9\cdot 5)$ (the denominator comes from the statement of the Bowtie Lemma).
By construction of $Z$, the same holds for each $(i,j) \in Z$.

Let $c_t$ be the constant from \autoref{proposition:degeneracy-clique-minor} and
\[d_3 = \left\lceil c_t \cdot k\sqrt{\log k}\right\rceil.\]
If, for some $F \subseteq Z \cup M_1 \cup M_2$ of size $1$ or $2$ the intersection graph $\Int(\bigcup_{g \in F}Z_g)$ is not $d_3$-degenerate, then immediately apply \autoref{lemma:dense-winning-scenario} to construct the desired bramble and the algorithm stops.
Thus, we now assume that for any such $F$ the graph $\Int(\bigcup_{g \in F}Z_g)$is $d_3$-degenerate.

Let
\begin{align*}
d_2 &= \left\lceil 2^9\cdot 5\cdot \left(e \cdot 4a^2 \cdot d_3\right)^{\alpha}\right\rceil \text{ and}\\
d_1 &= \left\lceil 2^9\cdot 5\cdot \left(e \cdot 4a^2 \cdot d_2\right)^{\alpha}\right\rceil.
\end{align*}

Let $J$ be the undirected $(|Z| + |M_1| + |M_2|)$-partite graph with partition $\{V_g \mid g \in Z \cup M_1 \cup M_2\}$ and such that, for every pair of distinct $g, g' \in Z \cup M_1 \cup M_2$, the graph $J[V_g \cup V_{g'}]$ is a copy of $\Int(Z_g \cup Z_{g'})$.
The goal now is to apply \autoref{theorem:poly-LLL-independent-sets} to select, from each $g \in Z \cup M_1 \cup M_2$, a walk $W_g \in Z_g$ such that all such walks are pairwise vertex-disjoint.

Since for any distinct $g, g' \in Z \cup M_1 \cup M_2$ the graph $\Int(Z_g)$ is $d_3$-degenerate and for each $g \in Z \cup M_1 \cup M_2$ we have $|Z_g| \geq d_2/(2^9\cdot 5)$, it follows that
\[\frac{d_2}{2^9 \cdot 5} \geq (e \cdot 4a^2 \cdot d_3)^{\alpha}.\]
Thus to apply \autoref{theorem:poly-LLL-independent-sets} we need that
\[\left((e \cdot 4a^2\cdot d_3)^{\alpha}\right)^{1-\varepsilon} \geq \left(e \cdot 4a^2 \cdot d_3\right)^{1+\varepsilon}.\]
Similarly to case \lipItem{1.}, it suffices to pick $\alpha$ and $\varepsilon$ such that $\alpha(1-\varepsilon) \geq (1+\varepsilon)$, $0 < \varepsilon < 1$, and $\alpha > 1$. 
From these walks, a path system is constructed and given as input to \autoref{lemma:sparse-winning-scenario} in order to construct a bramble of congestion at most $6$, and the proof of case \lipItem{2.} is finished.

For case \lipItem{3.}, the goal is to apply \autoref{lemma:sparse-winning-scenario-wrapped}.
Let $\mathcal{I} = |V \setminus V(M_2)|$.
By assumption, $|\mathcal{I}| \geq 0.6a(a-1)$ and thus the first condition of \autoref{lemma:sparse-winning-scenario-wrapped} is satisfied.
As a consequence of the applications of the Bowtie Lemma and the choice of $Z$, the authors of~\cite{Masarik2022} construct, for each $(i,j) \in \mathcal{I}$, a linkage $L^{\mathfrak{Z}}_{i,j}$ such that $|L^{\mathfrak{Z}}_{i,j}| \geq d_1/(2^9 \cdot 5)$ and for each pair of distinct $(i,j),(i',j') \in \mathcal{I}$ the graph $\Int(L^{\mathfrak{Z}}_{i,j} \cup L^{\mathfrak{Z}}f_{i',j'})$ is $d_2$-degenerate.
Thus, similarly to the previous cases, to apply \autoref{lemma:sparse-winning-scenario-wrapped} we need that
\[\left((e \cdot 4a^2\cdot d_2)^{\alpha}\right)^{1-\varepsilon} \geq \left(e \cdot 4a^2 \cdot d_2\right)^{1+\varepsilon}\]
and the previous choices of $\alpha$ and $\varepsilon$ also suffice in this case.

For reference, we now repeat the choices of $b, d_1, d_2$, and $d_3$.
We also include the choice of $a$ as it is made in~\cite{Masarik2022}.
Let $c$ be the constant of \autoref{lemma:sparse-winning-scenario-wrapped} and $c_a = c^{-4}$.
In this article, we chose
\begin{align*}
a &= \left\lceil c_a \cdot k^2 \sqrt{1+ \log k}\right\rceil,\\
d_3 &= \left\lceil c_t \cdot k\sqrt{\log k}\right\rceil,\\
d_2 &= \left\lceil 2^9\cdot 5\cdot \left(e \cdot 4a^2 \cdot d_3\right)^{\alpha}\right\rceil,\\
d_1 &= \left\lceil 2^9\cdot 5\cdot \left(e \cdot 4a^2 \cdot d_2\right)^{\alpha}\right\rceil, \text{ and}\\
b &= \left\lceil (e \cdot 4a^2 \cdot d^2_1)^{\alpha}\right\rceil.
\end{align*}

Since we started with an $(a^2,b^2)$-path system, the choices of $t$ and $\ell$ in the statement of \autoref{theorem:constant-congestion-brambles-poly-time} are $a^2$ and $b^2$, respectively.

\subsection{Analysis of the running time}\label{subsection:running-time-analysis}
The main contributions to the running time of the algorithm come from applications of \cite[Lemmas 3.5, 3.6, and 3.7]{Masarik2022} and \autoref{theorem:poly-LLL-independent-sets}, which depends on the LLL.
The other constructions and selections done internally to the proof presented by Masa\v{r}\'{\i}k et al. in \cite[Section 4]{Masarik2022} can be done by simply following paths of the path system and picking walks and/or paths with special properties that are guaranteed to exist from the existential proofs and that can be efficiently verified.
Thus, the impact of such steps on the running time is overshadowed by the contribution of the aforementioned lemmas, and we refrain from commenting on those in order to maintain a reasonable size for this article.

Let $n = |V(D)|$.
As remarked above, \cite[Lemmas 3.5 and 3.6]{Matthias2014} are both applied to all vertices of $G$ to build closed walks and (untangled) threaded linkages, and both naturally yield polynomial-time algorithms.
The first can be realized by simply traversing the paths of each linkage $L_{i,j}$ of the $(a,b)$-path system in a particular order that guarantees that each  $P \in L_{i,j}$ is traversed only once.
The resulting algorithm runs in time $\mathcal{O}(a^2 \cdot b \cdot n)$ since every $L_{i,j}$ has size $b$, $a(a-1)$ linkages are processed, and each path has size at most $n$.
Similarly, the latter of those two lemmas yields an algorithm that runs in time $\mathcal{O}(b \cdot n)$ by following the greedy construction traversing the walk associated with the threaded linkage of size $b$.
Thus the applications of \cite[Lemmas 3.5 and 3.6]{Matthias2014} contribute to a total of $\mathcal{O}(a^2 \cdot b \cdot n)$ to the running time of the construction.

For $i \in [2]$, the graph $H_i$ defined in the construction is built by constructing intersection graphs and testing their degeneracy.
The intersection graph of two linkages of size $b$ can be built in time $\mathcal{O}(b^2 \cdot n)$.
Matula and Beck~\cite{10.1145/2402.322385} showed that the degeneracy of a graph can be found in time $\mathcal{O}(n^2)$\footnote{The actual running time is $\mathcal{O}(n+m)$, where $m$ is the number of edges of the input graph.}.
Each of these procedures is applied $\binom{a(a-1)}{2}$ times since each vertex of $G$ is associated with a linkage of the given $(a,b)$-path system, and thus $|V(G)| = a(a-1)$.
The construction of each $H_i$ therefore contributes with $\mathcal{O}(a^4 \cdot (b^2\cdot n + n^2))$ to the running time of the construction.

The running time of the Bowtie Lemma~\cite[Lemma 3.7]{Masarik2022} depends on the running time of the Partitioning Lemma of Masa\v{r}\'{\i}k et al.~\cite[Lemma 11]{MasarikPacking2022} (see also cf.~\cite[Lemma 3.8]{Masarik2022}).
The latter consists of counting of degrees of sequences of vertices of a bipartite graph to build subgraphs with a guaranteed number of edges, and proceeding recursively on those subgraphs.
The procedure yields an algorithm running in time $\mathcal{O}(\ell \cdot n^2)$, where $\ell$ is, informally, the number of parts we want to construct.
After this, the remaining part of the Bowtie Lemma is algorithmic, running in time $\mathcal{O}(b^2\cdot n^2)$, since the input threaded linkages have size $b$, and thus the applications of the Bowtie Lemma contribute to a total of $\mathcal{O}(a^4(d_1 \cdot n^2 + b^2 \cdot n^2))$ since it is applied to threaded linkages associated with edges of $M_1$ and $M_2$.

Up to this point, the total running time includes the terms
\[\mathcal{O}(a^2 \cdot b \cdot n) + \mathcal{O}(a^4 \cdot (b^2\cdot n + n^2)) + \mathcal{O}(a^4(d_1 \cdot n^2 + b^2 \cdot n^2)),\]
which boils down to $\mathcal{O}(a^4\cdot b^2\cdot n^2)$ since $b > d_1$.
From the choice of $b$, we get
\[b = \mathcal{O}\left(k^{4\alpha + 8\alpha^2 + 10\alpha^3} \cdot (\log k)^{\alpha + 2\alpha^2 + 3\alpha^3}\right)\]
and thus so far the running time is bounded by
\begin{equation}\label{eq:running-time-first-part}
\mathcal{O}\left(k^{8 + 8\alpha + 16\alpha^2 + 20\alpha^3} \cdot (\log k)^{2 + 2\alpha  + 4\alpha^2 + 6\alpha^3} \cdot n^2\right).
\end{equation}
since $a^4 = \Ocal(k^8 \cdot   \log^2 k)$.

For $\alpha = 1$ we can apply \autoref{proposition:Harris-independent-set} to obtain \autoref{theorem:strongly-connected-implies-linkedness-using-harris}.

To guarantee that the construction ends in polynomial time when using \autoref{theorem:poly-LLL-independent-sets}, the relationship $\alpha(1-\varepsilon) \geq 1 + \varepsilon$ must be respected.
Since $\alpha \geq 1$ is needed for the construction to be possible, this implies that $0 < \varepsilon <  1$ and $\alpha > 1$. Together with the fact that $\varepsilon$ is inversely proportional to the running time stated in the \autoref{proposition:LLL-polynomial-time}, these observations imply a trade-off on the choices of $\varepsilon$ and $\alpha$.
In other words, the closer $\varepsilon$ is to $1$, the faster the construction ends, and the price that we pay is increasing the dependency on $k$ that is used to find the bramble, which in turn depends on $\alpha$.

When running the algorithm, \autoref{theorem:poly-LLL-independent-sets} is applied at most once since it appears once in the three cases distinguished in the proof.
It is either used as part of the construction of the path system used to solve case \lipItem{2.}, or is implicitly applied in a call to \autoref{lemma:sparse-winning-scenario-wrapped} when solving cases \lipItem{1.} and \lipItem{3.}.
The running time of \autoref{theorem:poly-LLL-independent-sets} is split into two parts: the construction of the input graph $G$ and then an application of \autoref{proposition:LLL-polynomial-time}.
When constructing the bramble, the graph $G$ on which \autoref{theorem:poly-LLL-independent-sets} is applied can be built by constructing at most $a^4$ intersection graphs of linkages of size $b$, and so $G$ can be built in time $\mathcal{O}(a^4 \cdot b^2 \cdot n)$.
The dependence on $n$ is needed only to verify if two given paths intersect on a vertex and has no impact on the size of $V(G)$, which turns out to be $\mathcal{O}(a^2 \cdot b)$ since it contains at most $a^2$ linkages of size $b$, and thus \autoref{proposition:LLL-polynomial-time} runs on an input whose size depends on $k$ and $\alpha$ only.

From the aforementioned discussions, we distinguish two parts determining the total running time of the construction, writing it as
\[\mathcal{O}\left(k^{8 + 8\alpha + 16\alpha^2 + 20\alpha^3} \cdot (\log k)^{2 + 2\alpha  + 4\alpha^2 + 6\alpha^3} \cdot n^2 + \poly(k, \alpha)\right),\]
where the first term of the running time comes from \autoref{eq:running-time-first-part} and the second term from \autoref{theorem:poly-LLL-independent-sets}, choosing $\varepsilon$ to be the largest possible while respecting the aforementioned restrictions.
Based on assisted computation, reasonable choices for $\alpha$ and $\varepsilon$ with respect to the term that is independent of $n$ are roughly $1.66$ and $0.248$, respectively.
These choices yield a running time of $\mathcal{O}(k^{156.76}\cdot (\log k)^{43.77} \cdot n^2 + k^{566.13}\cdot (\log k)^{189.66})$ and $t_k = \mathcal{O}(k^{301.21}\cdot (\log k^2)^{42.15})$ in the statement of \autoref{theorem:strongly-connected-implies-linkedness}.
The impact of the first term approaches $\mathcal{O}(k^{52} \cdot (\log k)^{14} \cdot n^2)$ as $\alpha$ approaches $1$.
With $\alpha = 1.001$, for example, we obtain $\mathcal{O}(k^{53.01} \cdot (\log k)^{14.29} \cdot n^2)$ for the first term, but the exponent on $k$ of the second term grows over $10000$.

\bibliography{main_full}

\end{document}